\documentclass[envcountsame]{llncs}
\usepackage{enumerate}
\usepackage{etoolbox}
\usepackage{stmaryrd}
\usepackage{amssymb,amsmath}
\usepackage{xspace}
\usepackage{mathtools}
\usepackage{microtype}
\usepackage{pifont}
\usepackage{tikz}
\usepackage{wrapfig}
\usepackage{graphicx,epsfig,color}

\usepackage[arrow,matrix,frame,curve,cmtip]{xy}\CompileMatrices

\usetikzlibrary{calc,arrows,fit,backgrounds}
\usetikzlibrary{shapes.geometric}
\usetikzlibrary{decorations.pathreplacing}

\makeatletter

\newcommand\@optsub[2]{
  \ifstrempty{#2}{%
    #1%
  }{%
    #1_{#2}%
  }%
}
\newcommand\@optsup[2]{
  {#1}%
  \ifstrempty{#2}{}{^{#2}}%
}
\newcommand\@optapp[2]{
  {#1}%
  \ifstrempty{#2}{}{(#2)}%
}

\renewcommand\phi{\varphi}
\renewcommand\emptyset{\varnothing}

\tikzset{
  mono/.style={>->},
  gnode/.style={circle,fill=black,inner sep=0mm,minimum size=2mm,font=\scriptsize,text=white},
  gedge/.style={->,>=latex},
  arlab/.style={inner sep=1pt,font=\scriptsize},
  glab/.style={inner sep=1pt,font=\scriptsize},
  hyperedge/.style={shape=rectangle,draw,inner sep=0,
    minimum width=1cm,minimum height=.4cm},
  point/.style={shape=circle,inner sep=0,fill=black,
    minimum height=1pt,minimum width=1pt}
}

\newcommand\graphboxthick[2][grbox]{
  \begin{pgfonlayer}{background}
    \node[fit=#2] (#1) {} ;
    \fill[black!20,rounded corners=2mm,postaction={draw,black}] 
          (#1.north west) -- (#1.north east) --
          (#1.south east) -- (#1.south west) -- cycle ;
  \end{pgfonlayer}
}

\newcommand\interfacebox[2][ifbox]{
  \begin{pgfonlayer}{background}
    \node[fit=#2] (#1) {} ;
    \draw[dashed,rounded corners=2mm] 
          (#1.north west) -- (#1.north east) --
          (#1.south east) -- (#1.south west) -- cycle ;
  \end{pgfonlayer}
}

\newcommand\substx[2]{
  ($ (0,0)!#1!(0,1) + (#2,0) $)
}

\newcommand\substy[2]{
  ($ (0,0)!#1!(1,0) + (0,#2) $)
}

\newcommand\tpl[1]{\langle #1 \rangle}

\newcommand\arright[1][]{%
 \ifstrempty{#1}%
 {\rightarrow}%
 {\mathbin{%
     \mathchoice%
     {\xrightarrow{#1}}%
     {\scalebox{.8}[1]{$\textstyle\relbar$}{\raisebox{.23ex}{$\scriptstyle #1$}}{\shortrightarrow}}%
     {\scalebox{.8}[1]{$\scriptstyle\relbar$}{\raisebox{.15ex}{$\scriptscriptstyle #1$}}{\shortrightarrow}}%
     {\scalebox{.8}[1]{$\scriptscriptstyle\relbar$}{\raisebox{.15ex}{$\scriptscriptstyle #1$}}{\shortrightarrow}}%
 }}
}

\newcommand\arleft[1][]{%
 \ifstrempty{#1}%
 {\leftarrow}%
 {\mathchoice%
   {\xleftarrow{#1}}
   {\mathbin{{\textstyle\shortleftarrow}{\raisebox{.23ex}{$\scriptstyle #1$}}\scalebox{.8}[1]{$\textstyle\relbar$}}}
   {\mathbin{{\scriptstyle\shortleftarrow}{\raisebox{.15ex}{$\scriptscriptstyle #1$}}\scalebox{.8}[1]{$\scriptstyle\relbar$}}}
   {\mathbin{{\scriptscriptstyle\shortleftarrow}{\raisebox{.15ex}{$\scriptscriptstyle #1$}}\scalebox{.8}[1]{$\scriptscriptstyle\relbar$}}}
 }
}

\newcommand\sLab[1][]{\@optsub{\mathit{lab}}{#1}}
\newcommand\sSrc[1][]{\@optsub{\mathit{src}}{#1}}
\newcommand\sTgt[1][]{\@optsub{\mathit{tgt}}{#1}}

\newcommand\fLab[2][]{\sLab[#1](#2)}
\newcommand\fSrc[2][]{\sSrc[#1](#2)}
\newcommand\fTgt[2][]{\sTgt[#1](#2)}

\newcommand\sFlower[1][]{\@optsub{\mbox{\ding{82}}}{#1}}
\newcommand\sFlowerM[1][]{\@optsub{\mathit{fl}}{#1}}

\newcommand{\mytilde}{{\raise.17ex\hbox{$\scriptstyle\mathtt{\sim}$}}\xspace}

\definecolor{dmagenta}{rgb}{0.81,0,0.81}
\definecolor{dcyan}{rgb}{0,0.6,0.6}
\definecolor{dgreen}{rgb}{0,0.6,0}
\definecolor{ltgray}{rgb}{.8,.8,.8}

\newcommand\Yes{{\color{dgreen}\ding{51}}}
\newcommand\No{{\color{red}\ding{55}}}

\makeatother

\newenvironment{proposition_for}[2]{\noindent{\bf Proposition~\ref{#1}#2}\it}{}
\newenvironment{lemma_for}[2]{\noindent{\bf Lemma~\ref{#1}#2}\it}{}

\usepackage{todonotes}
 
\newcommand{\TGL}{\ensuremath{{\mathit{TGL}}}}
\newcommand{\GR}[1][]{{\ensuremath{\mathbf{Graph}}}}
\newcommand{\GRL}{{\ensuremath{\mathbf{Gr_\Lambda}}}}
\newcommand{\core}[1]{\ensuremath{\mathit{core}(#1)}}
\newcommand{\LMon}{{\ensuremath{\mathbf{{\ell}Mon}}}}
\newcommand{\mor}{\mathrel{\raisebox{1.3ex}{\scalebox{1}[-1]{\mbox{$\looparrowright$}}}}}

\newcommand{\short}[1]{}
\newcommand{\full}[1]{#1}

\title{Specifying Graph Languages with Type Graphs}
\author{Andrea Corradini\inst{1} \and Barbara K\"onig\inst{2} \and Dennis 
Nolte\inst{2}}
\institute{%
  Universit\`a di Pisa, Italy\\
  \email{andrea@di.unipi.it}
  \and Universit\"at Duisburg-Essen, Germany\\ 
    \email{\{barbara\_koenig,dennis.nolte\}@uni-due.de}
}
\begin{document}
\maketitle

\begin{abstract}
  We investigate three formalisms to specify graph languages, i.e.
  sets of graphs, based on type graphs. First, we are interested in
  (pure) type graphs, where the corresponding language consists of all
  graphs that can be mapped homomorphically to a given type graph. In
  this context, we also study languages specified by restriction
  graphs and their relation to type graphs. Second, we extend this
  basic approach to a type graph logic and, third, to type graphs with
  annotations. We present decidability results and closure properties
  for each of the formalisms.
\end{abstract}


\section{Introduction}
\label{sec:introduction}

Formal languages in general and regular languages in particular play
an important role in computer science. They can be used for pattern
matching, parsing, verification and in many other domains. For
instance, verification approaches such as reachability checking,
counterexample-guided abstraction refinement
\cite{cgjlv:abstraction-refinement-journal} and non-termination
analysis \cite{ez:non-termination} could be directly adapted to graph
transformation systems if one had a graph specification formalism with
suitable closure properties, computable pre- and postconditions and
inclusion checks. Inclusion checks are also important for
checking when a fixpoint iteration sequence stabilizes.

While regular languages for words and trees are well-understood and
can be used efficiently and successfully in applications, the
situation is less satisfactory when it comes to graphs. Although the
work of Courcelle \cite{ce:graph-structure-mso} presents an accepted
notion of recognizable graph languages, equivalent to regular
languages, this is often not useful in practice, due to
the sheer size of the resulting graph automata. Other formalisms, such
as application conditions
\cite{r:representing-fol,hp:nested-constraints} and first-order or
second-order logics, feature more compact descriptions, but there are
problems with expressiveness, undecidability issues or unsatisfactory
closure properties.\footnote{A more detailed overview over related
  formalisms is given in the conclusion
  (Section~\ref{sec:conclusion}).}

Hence, we believe that it is important to study and compare
specification formalisms (i.e., automata, grammars and logics) that allow
to specify potentially infinite sets of graphs. In our opinion there
is no one-fits-all solution, but we believe that specification
mechanisms should be studied and compared more extensively.

In this paper we study specification formalisms based on type graphs,
where a type graph $T$ represents all graphs that can be mapped
homomorphically to $T$, potentially taking into account some extra
constraints. Type graphs are common in graph rewriting
\cite{cmr:graph-processes,r:gra-handbook}. Usually, one assumes that
all items, i.e., rules and graphs to be rewritten, are typed,
introducing constraints on the applicability of rules. Hence, type
graphs are in a way seen as a form of labelling. This is different
from our point of view, where graphs (and rules) are -- a priori --
untyped (but labeled) and type graphs are simply a means to represent
sets of graphs.

There are various reasons for studying type graphs: first, they are
reasonably simple with many positive decidability results and they
have not yet been extensively studied from the perspective of
specification formalisms. Second, other specification mechanisms --
especially those used in connection with verification and abstract
graph transformation
\cite{r:canonical-graph-shapes,srw:shape-analysis-3vl,sww:abstract-gts}
-- are based on type graphs: abstract graphs are basically type graphs
with extra annotations. Third, while not being as expressive as
recognizable graph languages, they retain a nice intuition from
regular languages: given a finite state automaton $M$ one can think of the
language of $M$ as the set of all string graphs that can be mapped
homomorphically to $M$ (respecting initial and final states).

We in fact study three different formalisms based on type graphs:
first, pure type graphs $T$, where the language consists simply of all
graphs that can be mapped to $T$. We also  discuss the connection
between type graph and restriction graph languages. Then, in order to
obtain a language with better boolean closure properties, we study
type graph logic, which consists of type graphs enriched with boolean
connectives (negation, conjunction, disjunction). Finally, we consider
annotated type graphs, where the annotations constrain the number of
items mapped to a specific node or edge, somewhat similar to the
proposals from abstract graph rewriting mentioned above.

In all three cases we are interested in closure properties and in
decidability issues (such decidability of the membership, emptiness
and inclusion problems) and in expressiveness.  \full{Proofs for all
  the results can be found in Appendix~\ref{sec:proofs}.}
\short{Proofs for all the results and an extended example for
  annotated type graphs can be found in
  \cite{ckn:specifying-graph-languages-arxiv}.  }
\section{Preliminaries}
\label{sec:preliminaries}

We first introduce graphs and graph morphisms. In the
context of this paper we use edge-labeled, directed graphs.

\begin{definition}[Graph]
  Let $\Lambda$ be a fixed set of edge labels. A
  \emph{$\Lambda$-labeled graph} is a tuple 
    $G = \tpl{V,E,\sSrc,\sTgt,\sLab}$,
  where 
    $V$ is a finite set of nodes, 
    $E$ is a finite set of edges, 
    $\sSrc,\sTgt\colon E\to V$ assign to each edge a source and
      a target node, and 
    $\sLab\colon E\to\Lambda$ is a labeling function.
\end{definition}

We will denote, for a given graph $G$, its components by $V_G$, $E_G$,
$\sSrc[G]$, $\sTgt[G]$ and $\sLab[G]$, unless otherwise indicated. 

\begin{definition}[Graph morphism]
  Let $G,G'$ be two $\Lambda$-labeled graphs. A \emph{graph morphism}
    $\varphi\colon G\to G'$
  consists of two functions 
    $\varphi_V\colon V_G\to V_{G'}$ and 
    $\varphi_E\colon E_G\to E_{G'}$,
  such that for each edge $e\in E_G$ it holds that 
    $\fSrc[G']{\varphi_E(e)} = \phi_V(\fSrc[G]{e})$,
    $\fTgt[G']{\varphi_E(e)} = \phi_V(\fTgt[G]{e})$ and 
    $\fLab[G']{\varphi_E(e)} = \fLab[G]{e}$.
  If $\varphi$ is both injective and surjective it is called an
  isomorphism.
\end{definition}

We will often drop the subscripts $V,E$ and write $\varphi$ instead of
$\phi_V$, $\phi_E$.  We will consider the category $\GR$ having
$\Lambda$-labeled graphs as objects and graph morphisms as arrows. The
set of its objects will be denoted by $\GRL$.
%
%
The categorical structure induces an obvious preorder on graphs,
defined as follows.


\begin{definition}[Homomorphism preorder]
  Given graphs $G$ and $H$, we write $G \to H$ if there is a graph
  morphism from $G$ to $H$ in $\GR$. The relation $\to$ is
  obviously a preorder (i.e.\ it is reflexive and transitive) and we
  call it the \emph{homomorphism preorder} on graphs. We write
  $G \nrightarrow H$ if $G \to H$ does not hold. Graphs $G$ and $H$
  are \emph{homomorphically equivalent}, written $G\sim H$, if
  both $G \to H$ and $H \to G$ hold.


\end{definition}


We will revisit the concept of \emph{retracts} and \emph{cores} from
\cite{nt:duality-finite-structures}.  \emph{Cores} are a convenient
way to minimize type graphs, as, according to
\cite{nt:duality-finite-structures}, all graphs $G,H$ with $G \sim H$
have isomorphic cores.

\begin{definition}[Retract and core]
  A graph $H$ is called a \emph{retract} of a graph $G$ if $H$ is a
  subgraph of $G$ and in addition there exists a
  morphism $\varphi \colon G \to H$. A graph $H$ is called a
  \emph{core} of $G$, written $H = \core{G}$, if it is a retract of $G$ and has itself no proper
  retracts.
\end{definition}

\begin{example}
  The graph $H$ is a retract of $G$, where the morphism $\varphi$ is indicated 
  by the node numbering:
  
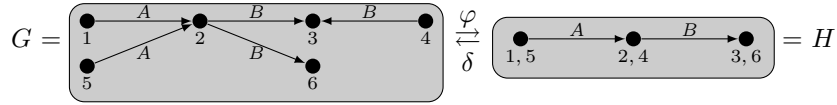
\begin{center}
  \begin{tabular}{ccccc}
    $G$ = &
    \begin{tikzpicture}[x=1.5cm,y=-0.6cm,baseline=(base.south)]
      \node[glab] (base) at (0,.5) {} ;
      \node[glab] (top) at (0,-.1) {} ;
      \node[gnode] (1) at (0,0) {} ; 
      \node[glab,below] (lab1) at (1.south) {$1$} ;
      \node[gnode] (2) at (1,0) {} ; 
      \node[glab,below] (lab2) at (2.south) {$2$} ;
      \node[gnode] (3) at (2,0) {} ; 
      \node[glab,below] (lab3) at (3.south) {$3$} ;
      \node[gnode] (4) at (3,0) {} ; 
      \node[glab,below] (lab4) at (4.south) {$4$} ;
      \node[gnode] (5) at (0,1) {} ; 
      \node[glab,below] (lab5) at (5.south) {$5$} ;
      \node[gnode] (6) at (2,1) {} ; 
      \node[glab,below] (lab6) at (6.south) {$6$} ;
      \draw[gedge] (1) to node[arlab,above] {$\mathit{A}$} (2) ;
      \draw[gedge] (5) to node[arlab,below] {$\mathit{A}$} (2) ;
      \draw[gedge] (2) to node[arlab,above] {$\mathit{B}$} (3) ;
      \draw[gedge] (2) to node[arlab,below] {$\mathit{B}$} (6) ;
      \draw[gedge] (4) to node[arlab,above] {$\mathit{B}$} (3) ;
      \graphboxthick[l]{ (1) (5) (4) (lab1) (lab5) (top)}
    \end{tikzpicture}
  & ${\begin{array}{c} \mbox{$\varphi$} \vspace{-0.1cm}\\ 
  \rightleftarrows\vspace{-0.1cm} \\\mbox{$\delta$} \end{array}}$
   &
    \begin{tikzpicture}[x=1.5cm,y=-1.2cm,baseline=(1.south)]
      \node[glab] (top) at (0,-.1) {} ;
      \node[gnode] (1) at (0,0) {} ; 
      \node[glab,below] (lab1) at (1.south) {$1,5$} ;
      \node[gnode] (2) at (1,0) {} ; 
      \node[glab,below] (lab2) at (2.south) {$2,4$} ;
      \node[gnode] (3) at (2,0) {} ; 
      \node[glab,below] (lab3) at (3.south) {$3,6$} ;
      \draw[gedge] (1) to node[arlab,above] {$\mathit{A}$} (2) ;
      \draw[gedge] (2) to node[arlab,above] {$\mathit{B}$} (3) ;
      \graphboxthick[r]{ (1) (3) (lab1) (lab3) (top)}
    \end{tikzpicture}
  & = $H$
  \end{tabular}
\end{center}
  Since the graph $H$ does not have a proper retract itself it is also the core 
  of $G$.
\end{example}

\section{Languages Specified by Type or Restriction Graphs}
\label{sec:type-graphs-new}

In this section we introduce two classes of graph languages that are
characterized by two somewhat dual properties. A \emph{type graph
  language} contains all graphs that can be mapped homomorphically to
a given \emph{type graph}, while a \emph{restriction graph language}
includes all graphs that \emph{do not contain} an homomorphic image of
a given \emph{restriction graph}.  Next, we discuss for these two
classes of languages some properties such as closure under set
operators, decidability of emptiness and inclusion, and decidability
of closure under rewriting via double-pushout rules. Finally we
discuss the relationship between these two classes of graph
languages.

\begin{definition}[Type graph language]
\label{def:typeGraphLanguage}
A \emph{type graph} $T$ is just a $\Lambda$-labeled graph. The language $\mathcal{L}(T)$ is 
defined as:
\[ \mathcal{L}(T) = \{ G \mid G \arright T \}. 
\] 
\end{definition}

  \begin{wrapfigure}{r}{3.1cm}
  \vspace{-1cm}
    \begin{tabular}{rl}
      $T_{\sFlower} = {}$ &
      \begin{tikzpicture}[x=1.5cm,y=-1.2cm,baseline=(1)]
        \node[gnode] (1) at (0,0) {} ;
        \draw[gedge] (1) .. controls +(-20:1cm) and +(20:1cm) .. 
        node[arlab,right] {A} (1) ;
        \draw[gedge] (1) .. controls +(114:1cm) and +(154:1cm) .. 
        node[arlab,above left] {B} (1) ;
        \draw[gedge] (1) .. controls +(216:1cm) and +(256:1cm) ..
        node[arlab,below right] {C} (1) ;
      \end{tikzpicture}
    \end{tabular}
    \vspace{-.9cm}
  \end{wrapfigure}

\begin{example}
\label{ex:typeGraph}
The following type graph $T$ over the edge label set
$\Lambda = \{A,B\}$ specifies a type graph language $\mathcal{L}(T)$
consisting of infinitely many graphs:
\begin{center}
  \begin{tabular}{ccccccccccc}
    {\LARGE $\mathcal{L}$(}
    \begin{tikzpicture}[x=1.5cm,y=-1.2cm,baseline=(1.south)]
      \node[gnode] (1) at (0,0) {} ; 
      \node[gnode] (2) at (.75,0) {} ;
      \draw[gedge] (1) to node[arlab,above] {$\mathit{B}$} (2) ;
      \draw[gedge] (1) .. controls +(160:.75cm) and +(200:.75cm) ..
            node[arlab,left] (loop1) {$\mathit{A}$} (1) ;
    \end{tikzpicture}
    {\LARGE )} = \Bigg\{&
  & {\LARGE $\emptyset$}
  & {\LARGE ,}
    \begin{tikzpicture}[x=1.5cm,y=-1.2cm,baseline=(1.south)]
      \node[gnode] (1) at (0,0) {} ; 
    \end{tikzpicture}
  & {\LARGE ,}
    \begin{tikzpicture}[x=1.5cm,y=-1.2cm,baseline=(1.south)]
      \node[glab] (top) at (0,-.1) {} ;
      \node[gnode] (1) at (0,0) {} ; 
      \node[gnode] (2) at (.75,0) {} ; 
      \node[gnode] (3) at (1.5,0) {} ; 
      \draw[gedge] (1) to node[arlab,above] {$\mathit{A}$} (2) ;
      \draw[gedge] (2) to node[arlab,above] {$\mathit{B}$} (3) ;
    \end{tikzpicture}
  & {\LARGE ,}
   &
    \begin{tikzpicture}[x=1.5cm,y=-1.2cm,baseline=(1.south)]
      \node[gnode] (1) at (0,0) {} ; 
      \node[gnode] (2) at (.75,0) {} ; 
	    \draw[gedge] (1) to[bend right=16] node[glab,below] 
	    (labA1) {$A$} (2) ;
	    \draw[gedge] (2) to[bend right=16] node[glab,above] 
	    (labA2) {$A$} (1) ;
    \end{tikzpicture}
  & {\LARGE ,}
  & {\LARGE \ldots}
  & \Bigg\}
  \end{tabular}
\end{center}
\end{example}

\noindent The category $\GR$ has a final object, that we denote
$T^{\Lambda}_{\sFlower}$, consisting of one node (called \emph{flower
  node} $\sFlower$) and one loop for each label in $\Lambda$.
Therefore $\mathcal{L}(T^{\Lambda}_{\sFlower}) = \GRL$. The graph
$T^{\Lambda}_{\sFlower}$ for $\Lambda = \{A,B,C\}$ is depicted to the
right.

\smallskip

Specifying graph languages using type graphs gives us the possibility
to forbid certain graph structures by not including them into the type
graph.  For example, no graph in the language of
Example~\ref{ex:typeGraph} can contain a $B$-loop or an $A$-edge
incident to the target of a $B$-edge.  However, it is not possible to
force some structures to exist in all graphs of the language, since
the morphism to the type graph need not be surjective. This point will
be addressed with the notion of \emph{annotated type graph} in
Section~\ref{sec:annot-tg}.

Another way (possibly more explicit) to specify languages of graphs not including certain 
structures, is the following one. 

\begin{definition}[Restriction graph language]
\label{def:restrictionGraphLanguage}
A \emph{restriction graph} $R$ is just a $\Lambda$-labeled graph. The language 
$\mathcal{L}_R(R)$ is 
defined as:
\[ \mathcal{L}_R(R) = \{ G \mid R \nrightarrow G\}. 
\] 
\end{definition}

We will consider the relationship between the class of languages
introduced in Definitions~\ref{def:typeGraphLanguage}
and~\ref{def:restrictionGraphLanguage} in
Section~\ref{sec:relatingTypeRestriction}. 

\subsection{Closure and Decidability Properties}
\label{sec:dppure-cppure}

The type graph and restriction graph languages enjoy the following 
complementary closure properties with respect to set operators.

\newcommand{\closurePropsPureAndRestrict}{ Type graph languages are
  closed under intersection (by taking the product of type graphs) but
  not under union or complementation, while restriction graph
  languages are closed under union (by taking the coproduct of
  restriction graphs) but not under intersection or complementation. }

\begin{proposition}
  \label{prop:closure properties of pure and restriction}
  \closurePropsPureAndRestrict
\end{proposition}

Instead the two classes of languages enjoy similar decidability properties.

\newcommand{\decidabilityPureAndRestrict}{ 
    For a graph language $\mathcal{L}$  characterized by a type graph 
    $T$ (i.e.\ $\mathcal{L} = \mathcal{L}(T)$) or by a restriction graph $R$ 
    (i.e.\ $\mathcal{L} = \mathcal{L}_R(R)$) 
   the following problems are decidable:
    \begin{enumerate}
      \item Membership, i.e.\ for each graph $G$ it is decidable if $G \in 
      \mathcal{L}$ holds.
      \item Emptiness, i.e.\ it is decidable if $\mathcal{L} = \emptyset$ holds.
    \end{enumerate}  
    Furthermore, language inclusion 
    is 
    decidable  for both classes of languages:
      \begin{enumerate}  
      \item[3.] Given type graphs $T_1$ and $T_2$ it is 
      decidable if $\mathcal{L}(T_1) \subseteq \mathcal{L}(T_2)$ holds.
      \item[4.] Given  restriction graphs $R_1$ and $R_2$ it is 
      decidable if $\mathcal{L}_R(R_1) \subseteq \mathcal{L}_R(R_2)$ holds.
    \end{enumerate} }

\begin{proposition}
  \label{prop:decidability problems of pure and restriction}
  \decidabilityPureAndRestrict
\end{proposition}

\subsection{Closure under Double-Pushout Rewriting}
\label{sec:icpure}

In this subsection we are using the DPO approach with general, not
necessarily injective, rules and matches.  We discuss how we can show
that a graph language $\mathcal{L}$ is a closed under a given graph
transformation rule
$\rho = (L \arleft[\varphi_L] I \arright[\varphi_R] R)$, i.e.,
$\mathcal{L}$ is an invariant for $\rho$. This means that for all
graphs $G,H$, where $G$ can be rewritten to $H$ via $\rho$, it holds
that $G \in \mathcal{L}$ implies $H \in \mathcal{L}$.

For both type graph languages and restriction graph languages,
separately, we characterize a sufficient and necessary condition which
shows that closure under rule application is decidable.  The condition
for restriction graph languages is related to a condition already
discussed in \cite{hw:consistency-conditional-gragra}.

\newcommand{\closureUnderRewritingRestrict}{ A restriction graph
  language $\mathcal{L}_R(S)$ is closed under a rule
  $\rho = (L \arleft[\varphi_L] I \arright[\varphi_R] R)$ if and only
  if the following condition holds: for every pair of morphisms
  $\alpha\colon R\to F$, $\beta\colon S\to F$ which are jointly
  surjective, applying the rule $\rho$ with (co-)match $\alpha$
  backwards to $F$ yields a graph $E$ with a homomorphic image of $S$,
  i.e., $E\not\in \mathcal{L}_R(S)$. }

\begin{proposition}[Closure under DPO rewriting for restriction graphs]
  \label{prop:closure-rg}
  \closureUnderRewritingRestrict
\end{proposition}

\newcommand{\closureUnderRewritingTypeGraphs}{ A type graph language
  $\mathcal{L}(T)$ is closed under a rule
  $\rho = (L \arleft[\varphi_L] I \arright[\varphi_R] R)$ if and only if for each
  morphism $t_L\colon L\to \core{T}$, there exists a morphism
  $t_R\colon R\to \core{T}$ such that
  $t_L\circ\varphi_L = t_R\circ\varphi_R$.
    
    \begin{center}
      \begin{tikzpicture}[shorten >=1pt, node distance=15mm and 15mm, on grid]
      	\draw
      		node (L) at (0,0) {\(L\)}
      		node [right=of L] (I) {\(I\)}
      		node [right=of I] (R) {\(R\)}
      		node [below=of I] (T) {\(\core{T}\)}
          node (equi) at (-1,-.7) {\(\Leftrightarrow\)}
          node (equi) at (-4.5,-.7) {$\mathcal{L}(T)$ is closed under application 
          of 
          $\rho$};
      
      	\begin{scope}[decoration={brace, raise=4mm}]
      		\draw[decorate] (L.west) -- node[midway, above=5mm] {\(\rho\)} 
      		(R.east);
      
      	\end{scope}	
      	\begin{scope}[->]
      		\draw (L) -- node[midway, left=2mm] {\(\forall t_L\)}  (T);
      		\draw (I) -- node[midway, above=.1mm] {\(\varphi_L\)} (L);
      		\draw (I) -- node[midway, above=.1mm] {\(\varphi_R\)}(R);
      		\draw (R) -- node[midway, right=2mm] {\(\exists t_R\)} (T);
      	\end{scope}
      \end{tikzpicture}
    \end{center}
  
    }

\begin{proposition}[Closure under DPO rewriting for type graphs]
  \label{prop:closure-tg}
  \closureUnderRewritingTypeGraphs
\end{proposition}


  We show that the \emph{only if part} ($\Rightarrow$) of
  Proposition~\ref{prop:closure-tg} cannot be weakened by considering morphisms 
  to the type graph  $T$, instead of to $\core{T}$.  
In fact, consider the following type
  graph $T$ and the 
  rule $\rho$:
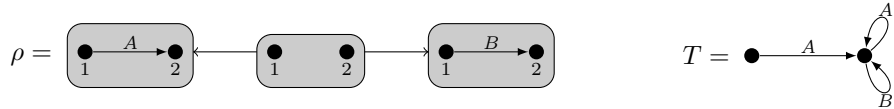
  \begin{center}
    \begin{tabular}{lcr}
      \begin{tabular}{rl}
        $\rho = {}$ &
        \begin{tikzpicture}[x=1.2cm,y=-1.2cm,baseline=(a1.south)]
      	  \begin{scope}[shift={(-2.5,0)}]
      	    \node[gnode] (a2) at ( 0,0) {} ; \node[glab,below] (lab1)
            at (a2.south) {$1$} ; \node[gnode] (a3) at ( 1,0) {} ;
            \node[glab,below] (lab3) at (a3.south) {$2$} ;
            \draw[gedge] (a2) -- node[glab,above] (labA2)
            {$\mathit{A}$} (a3) ; \graphboxthick[l]{ (lab1) (a2) (a3)
              (lab3) (labA2)}
      	  \end{scope}
      	  \begin{scope}
      	    \node[gnode] (a1) at (-.4,0) {} ; \node[glab,below] (lab1)
            at (a1.south) {$1$} ; \node[gnode] (a3) at (.4,0) {} ;
            \node[glab,below] (lab3) at (a3.south) {$2$} ;
            \graphboxthick[i]{(a1) (lab1) (a3) (lab3)}
      	  \end{scope}
      	  \begin{scope}[shift={(2.5,0)}]
      	    \node[gnode] (a1) at (-1,0) {} ; \node[glab,below] (lab1)
            at (a1.south) {$1$} ; \node[gnode] (a2) at ( 0,0) {} ;
            \node[glab,below] (lab3) at (a2.south) {$2$} ;
            \draw[gedge] (a1) -- node[glab,above] (labA1)
            {$\mathit{B}$} (a2) ; \graphboxthick[r]{(a1) (lab1) (a2)
              (lab3) (labA1)}
      	  \end{scope}
      	  \draw[->] \substy{(i.west)}{0} -- \substy{(l.east)}{0} ;
          \draw[->] \substy{(i.east)}{0} -- \substy{(r.west)}{0} ;
        \end{tikzpicture}
      \end{tabular}
      & \qquad \qquad \quad &
      \begin{tabular}{rl}
        $T = {}$ &
        \begin{tikzpicture}[x=1.5cm,y=-1.2cm,baseline=(1.south)]
          \node[gnode] (1) at (0,0) {} ; \node[gnode] (2) at (1,0) {}
          ; \draw[gedge] (2) .. controls +(40:.75cm) and +(80:.75cm)
          ..  node[arlab,above] {$\mathit{A}$} (2) ; \draw[gedge] (2)
          .. controls +(280:.75cm) and +(320:.75cm) ..
          node[arlab,below] {$\mathit{B}$} (2) ; \draw[gedge] (1) to
          node[arlab,above] {$\mathit{A}$} (2) ;
        \end{tikzpicture}
      \end{tabular}
    \end{tabular}
  \end{center}

  The type graph $T$ contains the flower node, i.e., it has
  $T^{\{A,B\}}_{\sFlower}$ as subgraph. This ensures that each graph
  $G$, edge-labeled over $\Lambda = \{A,B\}$, is in the language
  $\mathcal{L}(T)$, and thus by rewriting any graph
  $G \in \mathcal{L}(T)$ into a graph $H$ using $\rho$ it is
  guaranteed that $H \in \mathcal{L}(T)$. However there is a morphism
  $t_L\colon L \to T$, the one mapping the $A$-labeled edge of $L$ to the
  left $A$-labeled edge of $T$, such that there exists no morphism
  $t_R \colon R \to T$ satisfying
  $t_L\circ\varphi_L = t_R\circ\varphi_R$.
  
%
\subsection{Relating Type graph and Restriction Graph Languages}
\label{sec:relatingTypeRestriction}

Both type graph and restriction graph languages specify collections of
graphs by forbidding the presence of certain structures. This is more
explicit with the use of restriction graphs, though. A natural
question is how the two classes of languages are related. A partial
answer to this is provided by the notion of \emph{duality pairs} and
by an important result concerning their existence, presented in
\cite{nt:duality-finite-structures}.\footnote{Note that in
  \cite{nt:duality-finite-structures} graphs are simple, but it can be
  easily seen that for our purposes the results can be transferred
  straightforwardly.}

\begin{definition}[Duality pair]
Given two graphs $R$ and $T$, we call $T$ the \emph{dual} of $R$ if for every graph 
$G$ it holds that $G \to T$ if and only if $R \nrightarrow G$. In this case the pair 
$(R,T)$ is called \emph{duality pair}. 
\end{definition}

Clearly, we have that $(R,T)$ is a duality pair if and only if the restriction graph language
$\mathcal{L}_R(R)$ coincides with the type graph language $\mathcal{L}(T)$.

\begin{example}
  Let $\Lambda = \{A,B\}$ be given. The following is a 
  duality pair:
\begin{center}
  \begin{tabular}{ccccc}
    $(R,T)$ = \Bigg(&
    \begin{tikzpicture}[x=1.5cm,y=-1.2cm,baseline=(1.south)]
      \node[glab] (top) at (0,-.1) {} ;
      \node[gnode] (1) at (0,0) {} ; 
      \node[glab,below] (lab1) at (1.south) {$1$} ;
      \node[gnode] (2) at (.75,0) {} ; 
      \node[glab,below] (lab2) at (2.south) {$2$} ;
      \node[gnode] (3) at (1.5,0) {} ; 
      \node[glab,below] (lab3) at (3.south) {$3$} ;
      \draw[gedge] (1) to node[arlab,above] {$\mathit{A}$} (2) ;
      \draw[gedge] (2) to node[arlab,above] {$\mathit{B}$} (3) ;
    \end{tikzpicture}
  & {\LARGE ,}
   &
    \begin{tikzpicture}[x=1.5cm,y=-1.2cm,baseline=(1.south)]
      \node[glab] (top) at (0,-.15) {} ;
      \node[gnode] (1) at (0,0) {} ; 
      \node[glab,below] (lab1) at (1.south) {$1$} ;
      \node[gnode] (2) at (.75,0) {} ; 
      \node[glab,below] (lab2) at (2.south) {$2$} ;
      \draw[gedge] (2) to node[arlab,above] {$\mathit{A,B}$} (1) ;
      \draw[gedge] (1) .. controls +(160:1cm) and +(200:1cm) ..
            node[arlab,left] (loop1) {$\mathit{A}$} (1) ;
      \draw[gedge] (2) .. controls +(-20:1cm) and +(20:1cm) ..
            node[arlab,right] (loop2) {$\mathit{B}$} (2) ;
    \end{tikzpicture}
  & \Bigg)
  \end{tabular}
\end{center}
%
%
\noindent
Since node 1 of $T$ is not the source of a $B$-labeled edge and node 2
is not the target of an $A$-labeled edge, for every graph $G$ we have
$G \to T$ iff it does not contain a node which is both the target of
an $A$-labeled edge and the source of a $B$-labeled edge. But it
contains such a node if and only if $R \to G$.
\end{example}

One can identify the class of restriction graphs for which a
corresponding type graph exists which defines the same graph
language. Results from \cite{nt:duality-finite-structures}
state\footnote{We refer to Lemma~2.3, Lemma~2.5 and Theorem~3.1 in
  \cite{nt:duality-finite-structures}.} that given a core graph $R$, a
graph $T$ can be constructed such that $(R,T)$ is a duality pair if
and only if $R$ is a tree.


Thus we have a precise characterisation of the intersection of the
classes of type and restriction graph languages: $\mathcal{L}$ belongs
to the intersection if and only if it is of the form
$\mathcal{L} = \mathcal{L}_R(R)$ and $\core{R}$ is a tree. It is worth
mentioning that the construction of $T$ from $R$ using the results
from \cite{nt:duality-finite-structures} contains two exponential
blow-ups. This can be interpreted by saying that type graphs have
limited expressiveness if used to forbid the presence of certain
structures.
%

\section{Type Graph Logic}
\label{sec:tg-logic}

In this section we investigate the possibility to define a language of
graphs using a logical formula over type graphs.  We start by defining
the syntax and semantics of a type graph logic ($\TGL$).  



\begin{definition}[Syntax and semantics of \TGL]
  A $\TGL$ formula $F$ over a fixed set of edge labels $\Lambda$ is
  formed according to the following grammar:
  \[
    F := T\ |\ F \lor F\ |\ F \land F\ |\ \lnot F, \qquad 
    \text{where $T$ is a type graph.}
  \]
\noindent
Each $\TGL$ formula $F$ denotes a graph language
$\mathcal{L}(F)\subseteq \GRL$ defined by structural induction as
follows:
%
\begin{align*}
  \mathcal{L}(T) &= \{G \in \GRL \mid G \to T\} & 
  \mathcal{L}(\lnot F) &= \GRL \setminus \mathcal{L}(F) \\
  \mathcal{L}(F_1 \land F_2) &= \mathcal{L}(F_1) 
  \cap \mathcal{L}(F_2) 
  & \mathcal{L}(F_1 \lor F_2) &= \mathcal{L}(F_1) \cup \mathcal{L}(F_2)
\end{align*}
\end{definition}

Clearly, due to the presence of boolean connectives, boolean closure
properties come for free.

\begin{example}
  Let the following $\TGL$ formula $F$ over 
  $\Lambda = 
  \{A,B\}$ be given:
  \begin{center}
    \begin{tabular}{ccc}
      {\large $F = \lnot$} &
      \begin{tikzpicture}[x=1.5cm,y=-1.2cm,baseline=(1.south)]
      	    \node[gnode] (a1) at (0,0) {} ;
      	    \draw[gedge] (a1) .. controls +(-20:1cm) and +(20:1cm) .. 
      	            node[arlab,right] (lab1) {A} (a1) ;
      \end{tikzpicture}
      {\Large $\land \ \lnot$} &
      \begin{tikzpicture}[x=1.5cm,y=-1.2cm,baseline=(1.south)]
      	    \node[gnode] (a1) at (0,0) {} ;
      	    \draw[gedge] (a1) .. controls +(-20:1cm) and +(20:1cm) .. 
      	            node[arlab,right] (lab1) {B} (a1) ;
      \end{tikzpicture}
    \end{tabular}
  \end{center}
  The graph language $\mathcal{L}(F)$ consists of all graphs
  which do not consist exclusively of $A$-edges or of
    $B$-edges, i.e., which contain at least one $A$-labeled edge and
  at least one $B$-labeled edge,
something that can not be expressed by pure type
  graphs.
\end{example}

We  now present some positive results for graph languages $\mathcal{L}(F)$ over $\TGL$ 
formulas $F$ with respect to decidability problems.
%
%
Due to the conjunction and negation operator, the emptiness (or
unsatisfiability) check is not as trivial as it is for pure type
graphs.
Note that thanks to the presence of boolean connectives,
inclusion can be reduced to emptiness.

\newcommand{\decidabilityTypeGraphLogic}{ 
  For a graph language $\mathcal{L}(F)$ characterized by a $\TGL$ formula 
  $F$, the following problems are decidable:
  \begin{itemize}
    \item Membership, i.e. for all graphs $G$ it is decidable if $G \in 
    \mathcal{L}(F)$ holds.
    \item Emptiness, i.e. it is decidable if $\mathcal{L}(F) = \emptyset$ holds.
    \item Language inclusion, i.e. given two $\TGL$ formulas $F_1$ and $F_2$ it 
    is decidable if $\mathcal{L}(F_1) \subseteq \mathcal{L}(F_2)$ holds.
  \end{itemize} }

\begin{proposition}
  \label{prop:decidability problems of the type graph logic}
  \decidabilityTypeGraphLogic
\end{proposition}

  Such a logic could alternatively also be defined based on
  restriction graphs. A related logic, for injective occurrences of
  restriction graphs, is studied in
  \cite{oep:logic-graph-constraints}, where the authors also give a
  decidability result via inference rules.






\section{Annotated Type Graphs}
\label{sec:annot-tg}

In this section we will improve the expressiveness of the type graphs
themselves, rather than using an additional logic to do so. We will
equip graphs with additional annotations. As explained in the
introduction, this idea was already used similarly in abstract graph
rewriting. In contrast to most other approaches, we will
investigate the problem from a categorical point of
view. 

The idea we follow is to annotate each element of a type graph with
pairs of multiplicities, denoting upper and lower bounds. We will
define a category of multiply annotated graphs, where we consider
elements of a lattice-ordered monoid (short $\ell$-monoid) as
multiplicities.

\begin{definition}[Lattice-ordered monoid]
  A lattice-ordered monoid ($\ell$-\linebreak monoid)
  $(\mathcal{M},+,\leq)$
  consists of a set $\mathcal{M}$, a partial order $\le$ and a binary
  operation $+$ such that
  \begin{itemize}
  \item $(\mathcal{M},\le)$ is a lattice.
  \item $(\mathcal{M},+)$ is a monoid; we denote its unit by $0$.
  \item It holds that $a+(b\lor c) = (a+b) \lor (a+c)$ and $a+(b\land
    c) = (a+b)\land (a+c)$, where $\land,\lor$ are the meet and join
    of $\le$.
  \end{itemize} 
  \noindent
  We denote by $\LMon$ the category having $\ell$-monoids as objects
  and as arrows monoid homomorphisms which are monotone. 
\end{definition}

\begin{example}
  \label{ex:lmonoid}
  Let $n \in \mathbb{N}\backslash\{0\}$ and take
  $\mathcal{M}_n = \{0,1,\dots,n,m\}$ (zero, one, $\dots$, $n$, many)
  with $0 \leq 1 \leq \dots\leq n\leq m$ and addition as monoid
  operation with the proviso that $\ell_1+\ell_2=m$ if the sum is
  larger than $n$. Clearly, for all $a,b,c \in \mathcal{M}_n$
  $a \lor b = \text{max}\{a,b\}$ and $a \land b = \text{min}\{a,b\}$.
  From this we can infer distributivity and therefore
  $(\mathcal{M}_n,+,\leq)$ forms an $\ell$-monoid.
  
  Furthermore, given a set $S$ and an $\ell$-monoid
  $(\mathcal{M},+,\leq)$, it is easy to check that also
  $(\{a\colon S \to \mathcal{M}\},+,\leq)$ is an $\ell$-monoid, where
  the elements are functions from $S$ to $\mathcal{M}$ and the partial
  order and the monoidal operation are taken pointwise.
  
  In the following we will sometimes denote an $\ell$-monoid by its
  underlying set.
\end{example}

%
\begin{definition}[Annotations and multiplicities for graphs]
  Given a functor $\mathcal{A}\colon \GR \to \LMon$, an
  \emph{annotation based on $\mathcal{A}$} for a graph $G$ is an
  element $a \in \mathcal{A}(G)$.  We write $\mathcal{A}_\phi$,
  instead of $\mathcal{A}(\phi)$, for the action of functor
  $\mathcal{A}$ on a graph morphism $\phi$.  We assume that for each
  graph $G$ there is a \emph{standard annotation} based on
  $\mathcal{A}$ that we denote by $s_G$, thus
  $s_G \in \mathcal{A}(G)$.
  
  \noindent Given an $\ell$-monoid $\mathcal{M}_n = \{0,1,\dots,n,m\}$
  we define the functor $\mathcal{B}^n:\GR \to \LMon$ as follows:
 \begin{itemize}
 \item for every graph $G$, $\mathcal{B}^n(G) = \{a\colon (V_G \cup E_G) \to 
 \mathcal{M}_n\}$;

\item for every graph morphism $\varphi \colon G \to G'$ and $a \in
  \mathcal{B}^n(G)$, we have \\
  $\mathcal{B}^n_{\varphi}(a) \colon V_{G'} \cup E_{G'} \to
  \mathcal{M}_n$ with:
  \[ \mathcal{B}^n_{\varphi}(a)(y) = \sum\limits_{\varphi(x)=y}^{} a(x), \quad 
  \textit{where } x \in (V_G \cup E_G) \textit{ and } y \in (V_{G'} \cup 
  E_{G'}) \]
\end{itemize}
\noindent Therefore an annotation based on a functor $\mathcal{B}^n$
associates every item of a graph with a number (or the top value
$m$). We will call such kind of annotations \emph{multiplicities}.
Furthermore, the action of the functor on a morphism transforms a
multiplicity by summing up (in $\mathcal{M}_n$) the values of all
items of the source graph that are mapped to the same item of the
target graph.

For a graph $G$, its \emph{standard multiplicity}  $s_G \in \mathcal{B}^n(G)$ 
is defined  as the function which maps every node and edge of $G$ to 1.
%
\end{definition}

Some of the results that we will present in the rest of the paper will
hold for annotations based on a generic functor $\mathcal{A}$, some 
only for annotations based on functors $\mathcal{B}^n$, i.e.~for multiplicities.

The type graphs which we are going to consider are enriched with a set
of pairs of annotations. The motivation for considering multiple
  annotations rather than a single one is mainly to ensure closure
  under union. Each pair can be interpreted as establishing a
lower and an upper bound to what a graph morphism can map to the
graph.

\begin{definition}[Multiply annotated graphs]
  Given a functor $\mathcal{A}\colon \GR \to \LMon$, a \emph{multiply
    annotated graph $G[M]$ (over $\mathcal{A}$)} is a graph $G$
  equipped with a finite set of pairs of annotations
  $M \subseteq \mathcal{A}(G)\times \mathcal{A}(G)$, such that
  $\ell \leq u$ for all $(\ell,u) \in M$.
  
  An arrow $\phi\colon G[M]\to G'[M']$, also called a \emph{legal
    morphism}, is a graph morphism $\phi\colon G\to G'$ such that for
  all $(\ell,u) \in M$ there exists $(\ell',u') \in M'$ with
  $\mathcal{A}_\phi(\ell)\ge \ell'$ and $\mathcal{A}_\phi(u) \le
  u'$. We will write $G[\ell,u]$ as an abbreviation of
  $G[\{(\ell,u)\}]$.  In case of annotations based on $\mathcal{B}^n$,
  we will often call a pair $(\ell, u)$ a \emph{double multiplicity}.
\end{definition}

Multiply annotated graphs and legal morphisms form a category.

\newcommand{\compOfLegalMorph}{ The composition of two legal morphisms
  is a legal
  morphism. }

\begin{lemma}
  \label{le:composition of legal morphism}
  \compOfLegalMorph
\end{lemma}

\begin{example}
  \label{ex:graphwdm}
  Consider the following multiply annotated graphs (over
  $\mathcal{B}^2$) $G[\ell,u]$ and $H[\ell',u']$, both having one
  double multiplicity.
  
  
  \begin{center}
    \begin{tabular}{ccccc}
      $G[\ell,u] = $ &
      \begin{tikzpicture}[x=1.5cm,y=-1.2cm,baseline=(1.south)]
      	    \node[gnode] (a1) at (0,0) {} ;
            \node[glab,below] (labm1) at (a1.south) {[1,1]} ;
      	    \node[gnode] (a2) at (1,0) {} ;
            \node[glab,below] (labm2) at (a2.south) {[1,m]} ;
      	    \draw[gedge] (a1) to
      	            node[arlab,above] (lab1) {A [0,1]} (a2) ;
      \end{tikzpicture}
      & \qquad \qquad \qquad & $H[\ell',u'] = $ &
      \begin{tikzpicture}[x=1.5cm,y=-1.2cm,baseline=(1.south)]
      	    \node[gnode] (a1) at (0,0) {} ;
            \node[glab,below] (labm1) at (a1.south) {[1,m]} ;
      	    \draw[gedge] (a1) .. controls +(-20:1cm) and +(20:1cm) .. 
      	            node[arlab,right] (lab1) {A [0,m]} (a1) ;
      \end{tikzpicture}
    \end{tabular}
  \end{center}
  
  As evident from the picture, multiplicities are represented by writing 
  the lower and upper bounds next to the corresponding graph elements.
  Note that there is a unique, obvious graph morphism $\varphi \colon G \to H$, mapping both
  nodes of $G$ to the only node of $H$. Concerning multiplicities, by adding the lower and upper
  bounds of the two nodes of $G$, one gets the interval $[2,m]$ which is included in the 
  interval of the node of $H$, $[1,m]$. Similarly, the double multiplicity $[0,1]$ of the edge of
  $G$ is included in $[0,m]$.
  Therefore, since both $\mathcal{B}^2_{\varphi}(\ell) \geq \ell'$ and
  $\mathcal{B}^2_{\varphi}(u) \leq u'$ hold, we can conclude that
  $\varphi \colon G[\ell,u] \to H[\ell',u']$ is a legal morphism.
\end{example}
%

%
%
%

We are
now ready to define how a graph language $\mathcal{L}(T[M])$ looks
like.

\begin{definition}[Graph languages of multiply annotated type graphs]
  We say that a graph $G$ is represented by a multiply annotated type 
  graph $T[M]$ whenever there exists a legal 
  morphism $\phi\colon G[s_G,s_G]\to T[M]$, i.e., there exists $(\ell,u) \in 
  M$ such that $\ell \le \mathcal{A}_\phi(s_G)\le u$. We will write $G\in 
  \mathcal{L}(T[M])$ in this case. Whenever $M = \emptyset$ for a multiply 
  annotated type graph $T[M]$ we get $\mathcal{L}(T[M]) = \emptyset$.
\end{definition}


\full{An extended example can be found in 
Appendix~\ref{sec:exampleAnnotGraphs}.}
\short{An extended example can be found in 
\cite{ckn:specifying-graph-languages-arxiv}.} 

\subsection{Decidability Properties for Multiply Annotated Graphs}
\label{sec:dpannot}

We now address some 
decidability problems for languages defined by multiply annotated graphs. We get positive results with
respect to the membership  and  emptiness problems. However,
for decidability of language inclusion  we only get partial results.

For the membership problem we can simply enumerate all graph morphisms
$\varphi \colon G \to T$ and check if there exists a legal morphism
$\varphi \colon G[s_G,s_G] \to T[M]$. 

The emptiness check is somewhat more involved, since we have to take
care of ``illegal'' annotations. 

\newcommand{\emptinessCheckAnnot}{ 
  For a graph language $\mathcal{L}(T[M])$ characterized by a multiply
  annotated type graph $T[M]$ over  $\mathcal{B}^n$
  the emptiness problem is decidable:
  $\mathcal{L}(T[M]) = \emptyset$ iff $M = \emptyset$ or for each
  $(\ell,u)\in M$ there exists an edge $e\in E_T$ such that
  $\ell(e) \ge 1$ and $(u(\sSrc(e)) = 0$ or $u(\sTgt(e)) = 0)$. }

\begin{proposition}
  \label{prop:emptiness is decidable for annotated type graphs}
  \emptinessCheckAnnot
\end{proposition}

Language inclusion can be deduced from the existence of a legal morphism 
between the two multiply annotated type graphs. 

\newcommand{\languageinclusionWithLegalMorphism}{ The existence of a
  legal morphism $\varphi \colon T_1[M] \to T_2[N]$ implies
  $\mathcal{L}(T_1[M]) \subseteq \mathcal{L}(T_2[N])$. }

\begin{proposition}
  \label{prop:language inclusion using legal morphisms}
  \languageinclusionWithLegalMorphism
\end{proposition}

We would like to remark that this condition is sufficient but not
necessary, and we present the following counterexample. Let the
following two multiply annotated type graphs $T_1[M_1]$ and $T_2[M_2]$
over $\mathcal{B}^1$ be given where $|M_1|=|M_2|=1$:

  \begin{center}
    \begin{tabular}{ccccc}
      $T_1[M_1] = $ &
      \begin{tikzpicture}[x=1.5cm,y=-1.2cm,baseline=(1.south)]
  	    \node[gnode] (a1) at (0,0) {} ;
   	    \node[glab,below] (lab1) at (a1.south) {$[1,m]$} ;
      \end{tikzpicture}
      & \qquad \qquad \qquad & $T_2[M_2] = $ &
      \begin{tikzpicture}[x=1.5cm,y=-1.2cm,baseline=(1.south)]
   	    \node[gnode] (b1) at (0,0) {} ;
   	    \node[glab,below] (labb1) at (b1.south) {$[1,1]$} ;
   	    \node[gnode] (b2) at (.5,0) {} ;
      	    \node[glab,below] (labb2) at (b2.south) {$[0,m]$} ;
      \end{tikzpicture}
    \end{tabular}
  \end{center}

  Clearly we have that the languages $\mathcal{L}(T_1[M_1])$ and
  $\mathcal{L}(T_2[M_2])$ are equal as both contain all discrete
  non-empty graphs. Thus
  $\mathcal{L}(T_1[M_1]) \subseteq \mathcal{L}(T_2[M_2])$, but there
  exists no legal morphism $\varphi \colon T_1[M_1] \to T_2[M_2]$. In
  fact, the upper bound of the first node of $T_2$ would be violated
  if the node of $T_1$ is mapped by $\varphi$ to it, while the lower
  bound would be violated if the node of $T_1$ is mapped to the other
  node.

\subsection{Deciding Language Inclusion for Annotated Type Graphs}
\label{sec:dpannot-incl}
In this section we show that if we allow only bounded graph languages consisting of graphs
up to a fixed pathwidth, the language inclusion problem becomes
decidable for annotations based on $\mathcal{B}^n$. Pathwidth is a well-known
concept from graph theory that intuitively measures how much a graph
resembles a path. 

The proof is based on the notion of recognizability, which will be
described via automaton functors that were introduced in
\cite{bk:rec-arrow-graph}. We start with the main result and explain
step by step the arguments that will lead to decidability.

\newcommand{\languageinclusionannot}{ The language inclusion problem
  is decidable for graph languages of bounded pathwidth  characterized by multiply
  annotated type graphs over $\mathcal{B}^n$. That is, given
  $k \in \mathbb{N}$ and two multiply annotated type graphs $T_1[M_1]$
  and $T_2[M_2]$ over $\mathcal{B}^n$, it is decidable whether
  $\mathcal{L}(T_1[M_1])^{\leq k} \subseteq
  \mathcal{L}(T_2[M_2])^{\leq k}$, where
  $\mathcal{L}(T[M])^{\leq k} = \{G \in \mathcal{L}(T[M]) \mid G$ has
  pathwidth $\leq k \}$.
}

\begin{proposition}
  \label{prop:language inclusion is decidable for multi annot graphs}
  \languageinclusionannot
\end{proposition}

Our automaton model, given by automaton functors, reads cospans
(i.e., graphs with interfaces) instead of single graphs. Therefore in
the following, the category under consideration will be
$\emph{Cospan}_m(\GR)$, i.e. the category of cospans of
graphs where the objects are discrete graphs $J,K$ and the arrows are
cospans $c \colon J \to G \leftarrow K$ where both graph morphisms are
injective. We will refer to the graph $J$ as the \emph{inner
  interface} and to the graph $K$ as the \emph{outer interface} of the
graph $G$. In addition we will sometimes abbreviate the cospan
$c \colon J \to G \leftarrow K$ to the short representation
$c \colon J \mor K$.




According to \cite{bbfk:tree-pathwidth-cospan-journal} a graph has
pathwidth $k$ iff it can be decomposed into cospans where each middle
graph of a cospan has at most $k+1$ nodes. Hence it is easy to check that a 
path has pathwidth~$1$, while a clique of order $k$ has pathwidth $k-1$.


Our main goal is to build an automaton which can read all graphs of
our language step by step, similar to the idea of finite automata
reading words in formal languages. Such an automaton can be
constructed for an unbounded language, where the pathwidth is not
restricted. However, we obtain a \emph{finite} automaton only if we
restrict the pathwidth. 
%
%
Then we can use well-known algorithms for finite automata to solve the
language inclusion problem. Note that, if we would use tree automata
instead of finite automata, our result could be generalized to graphs
of bounded \emph{treewidth}.

We will first introduce the notion of automaton functor (which is a
categorical automaton model for so-called recognizable arrow
languages) and which is inspired by Courcelle's theory of recognizable
graph languages \cite{ce:graph-structure-mso}.

\begin{definition}[Automaton functor \cite{bk:rec-arrow-graph}]
  An automaton functor \linebreak
  $\mathcal{C} \colon \emph{Cospan}_m(\GR) \to \mathbf{Rel}$ is a
  functor that maps every object $J$ (i.e., every discrete graph) to a
  finite set $\mathcal{C}(J)$ (the set of states of $J$) and every
  cospan $c \colon J \mor K$ to a relation
  $\mathcal{C}(c) \subseteq \mathcal{C}(J) \times \mathcal{C}(K)$ (the
  transition relation of $c$). In addition there is a distinguished
  set of initial states $I \subseteq \mathcal{C}(\emptyset)$ and a
  distinguished set of final states
  $F \subseteq \mathcal{C}(\emptyset)$.
  The \emph{language $\mathcal{L}_{\mathcal{C}}$} of
  $\mathcal{C}$ is defined as follows:
  \begin{quote}
    A graph $G$ is contained in $\mathcal{L}_{\mathcal{C}}$ if and
    only if there exist states $q \in I$ and $q' \in F$  which are related by
    $\mathcal{C}(c)$, i.e.~ $(q,q') \in\mathcal{C}(c)$,   where
    $c \colon \emptyset \to G \leftarrow \emptyset$ is the unique
    cospan with empty interfaces and middle graph $G$.
  \end{quote}
  Languages accepted by automaton functors are called
  \emph{recognizable}.

\end{definition}

We will now define an automaton functor for a type graph $T[M]$ over $\mathcal{B}^n$.

\begin{definition}[Counting cospan automaton] 
  Let $T[M]$ be a multiply annotated type graph over
  $\mathcal{B}^n$. We define an automaton functor
  $\mathcal{C}_{T[M]} \colon \emph{Cospan}_m(\GR)\linebreak \to
  \mathbf{Rel}$ as follows:
  \begin{itemize}
    \item For each object $J$ of $\emph{Cospan}_m(\GR)$ (thus $J$ is a finite discrete graph),  $\mathcal{C}_{T[M]}(J) = \{ (f,b) \mid f 
    \colon J \to T, b \in \mathcal{B}^n(T)\}$ is its finite set of states 
    \item $I \subseteq \mathcal{C}_{T[M]}(\emptyset)$ is the set of 
    \emph{initial} 
    states with $I = \{ (f \colon \emptyset \to T , 0) \}$, where $0$
    is the constant $0$-function 
    \item $F \subseteq \mathcal{C}_{T[M]}(\emptyset)$ is the set of 
    \emph{final} states 
    with $F = \{ (f \colon \emptyset \to T , b) \mid \exists (\ell,u) \in M 
    : \ell \leq b \leq u \}$ 
  \end{itemize}

  \noindent{}
    \begin{minipage}{.64\linewidth}
      {\textbf{\,--}}~Let $c \colon J \arright[\psi_L] G \arleft[\psi_R] K$ be an arrow in
      the category $\emph{Cospan}_m(\GR)$ with discrete interface
      graphs $J$ and $K$ where both graph morphisms
      $\psi_L \colon J \to G$ and $\psi_R \colon K \to G$ are
      injective.  Two states $(f \colon J \to T,b)$ and
      $(f' \colon K \to T,b')$ are in the relation
      $\mathcal{C}_{T[M]}(c)$ if and only if there exists a morphism
      $h \colon G \to T$ such that the diagram to the right commutes
      and for all $x \in V_T\cup E_T$ the following equation holds:
    \end{minipage}
    \quad
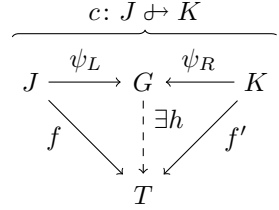
    \begin{minipage}{.26\linewidth}
      \begin{center}
        \begin{tikzpicture}[shorten >=1pt, node distance=15mm and 15mm, on grid]
        	\draw
        		node (J) at (0,0) {\(J\)}
        		node [right=of J] (G) {\(G\)}
        		node [right=of G] (K) {\(K\)}
        		node [below=of G] (T) {\(T\)};
        
        	\begin{scope}[decoration={brace, raise=6mm}]
        		\draw[decorate] (J.west) -- node[midway, above=7mm] {\(c \colon J 
        		\mor K\)} (K.east);
        
        	\end{scope}	
        	\begin{scope}[->]
        		\draw (J) -- node[midway, left=2mm] {\(f\)}  (T);
        		\draw (J) -- node[midway, above=.1mm] {\(\psi_L\)} (G);
        		\draw[dashed] (G) -- node[above right=.2mm] {\(\exists h\)} (T);
        		\draw (K) -- node[midway, above=.1mm] {\(\psi_R\)}(G);
        		\draw (K) -- node[midway, right=2mm] {\(f'\)} (T);
        	\end{scope}
        \end{tikzpicture}
      \end{center}
    \end{minipage}

    $$ b'(x) = b(x) + | \{ y \in (G\setminus \psi_R(K)) \mid h(y)=x \}| $$ 
    
    \noindent The set $G\setminus \psi_R(K)$ consists of all elements of
    $G$ which are not targeted by the morphism $\psi_R$, e.g.
    $G\setminus \psi_R(K) = (V_G \setminus \psi_R(V_K)) \cup (E_G \setminus
    \psi_R(E_K))$. 
  Instead of $\mathcal{L}_{\mathcal{C}_{T[M]}}$ and $\mathcal{C}_{T[M]}$ we 
  just write $\mathcal{L}_{\mathcal{C}}$ and $\mathcal{C}$ if  $T[M]$ 
  is clear from the context.
\end{definition}

The intuition behind this construction is to count for each item $x$ of $T$, step by step, the
number of elements that are being mapped from a graph $G$ (which is in
the form of a cospan decomposition) to $x$,  and then check if the bounds of a
pair of annotations $(\ell,u) \in M$ of the multiply annotated type
graph $T[M]$ are satisfied. We give a short example before moving on
to the results.

\begin{example}
  Let the following multiply annotated type graph (over
  $\mathcal{B}^2$) $T[\ell,u]$ and the cospan
  $(c \colon \emptyset \to G \leftarrow \emptyset)$ with $G \in
  \mathcal{L}(T[\ell,u])$ be given:
  
  \begin{center}
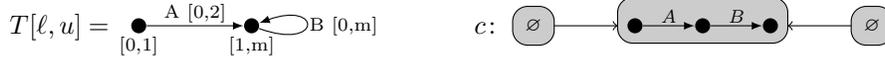
    \begin{tabular}{ccccc}
      $T[\ell,u] = $ &
      \begin{tikzpicture}[x=1.5cm,y=-1.2cm,baseline=(1.south)]
      	    \node[gnode] (a1) at (0,0) {} ;
            \node[glab,below] (labm1) at (a1.south) {[0,1]} ;
      	    \node[gnode] (a2) at (1,0) {} ;
            \node[glab,below] (labm2) at (a2.south) {[1,m]} ;
      	    \draw[gedge] (a1) to
      	            node[arlab,above] (lab1) {A [0,2]} (a2) ;
      	    \draw[gedge] (a2) .. controls +(-20:1cm) and +(20:1cm) .. 
      	            node[arlab,right] (labm3) {B [0,m]} (a2) ;
      \end{tikzpicture}
      & \qquad \quad \quad & 
            $c \colon {}$ &
      \begin{tikzpicture}[x=.9cm,y=-1.2cm,baseline=(a1.south)]
        \begin{scope}[shift={(-2.5,0)}]
          \node[glab] (lab1) at (0,0) {$\emptyset$} ;
          \graphboxthick[j]{ (lab1) }  
        \end{scope}
        \begin{scope}
          \node[gnode] (a1) at (-1,0) {} ;
          \node[gnode] (a2) at (0,0) {} ;
          \node[gnode] (a3) at (1,0) {} ;
          \draw[gedge] (a1) -- node[glab,above] (labA1) {$\mathit{A}$} (a2) ;
          \draw[gedge] (a2) -- node[glab,above] (labB1) {$\mathit{B}$} (a3) ;
          \graphboxthick[g]{(a1) (a2) (a3) (labA1) (labB1)}  
        \end{scope}
        \begin{scope}[shift={(2.5,0)}]
          \node[glab] (lab1) at (0,0) {$\emptyset$} ;
          \graphboxthick[k]{ (lab1) }  
        \end{scope}
        \draw[->] \substy{(j.east)}{0} -- \substy{(g.west)}{0} ;
        \draw[->] \substy{(k.west)}{0} -- \substy{(g.east)}{0} ;
      \end{tikzpicture}
    \end{tabular} 
  \end{center} 

\noindent We will now decompose the cospan $c$ into two cospans $c_1,c_2$ with 
$c = c_1;c_2$ in the following way:
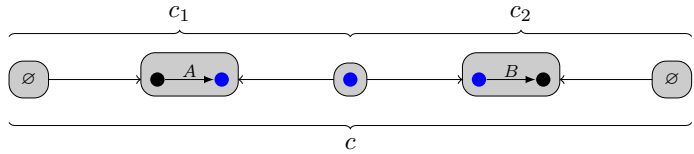
\begin{center}
\vspace*{-.2cm}
  \scalebox{.95}{
  \begin{tikzpicture}[x=.9cm,y=-1.2cm,baseline=(a1.south)]
    \begin{scope}[shift={(-2.5,0)}]
      \node[glab] (lab1) at (0,0) {$\emptyset$} ;
      \graphboxthick[l]{ (lab1) }  
    \end{scope}
    \begin{scope}
      \node[gnode] (a1) at (-.5,0) {} ;
      \node[gnode,fill=blue] (a2) at (.5,0) {} ;
      \draw[gedge] (a1) -- node[glab,above] (labA1) {$\mathit{A}$} (a2) ;
      \graphboxthick[i]{(a1) (a2) (labA1)}  
    \end{scope}
    \begin{scope}[shift={(2.5,0)}]
      \node[gnode,fill=blue] (a1) at (0,0) {};
      \graphboxthick[r]{ (a1) }  
    \end{scope}
    \begin{scope}[shift={(5,0)}]
      \node[gnode,fill=blue] (a2) at (-.5,0) {} ;
      \node[gnode] (a3) at (.5,0) {} ;
      \draw[gedge] (a2) -- node[glab,above] (labB1) {$\mathit{B}$} (a3) ;
      \graphboxthick[g]{ (a2) (a3) (labB1)}  
    \end{scope}
    \begin{scope}[shift={(7.5,0)}]
      \node[glab] (lab1) at (0,0) {$\emptyset$} ;
      \graphboxthick[k]{ (lab1) }  
    \end{scope}
    \draw[->] \substy{(l.east)}{0} -- \substy{(i.west)}{0} ;
    \draw[->] \substy{(r.west)}{0} -- \substy{(i.east)}{0} ;
    \draw[->] \substy{(r.east)}{0} -- \substy{(g.west)}{0} ;
    \draw[->] \substy{(k.west)}{0} -- \substy{(g.east)}{0} ;
    \begin{scope}[decoration={brace, raise=6mm}]
      \draw[decorate] (l.west) -- node[midway, above=7mm] {\(c_1\)} (r.center);
     \end{scope}
    \begin{scope}[decoration={brace, raise=6mm}]
      \draw[decorate] (r.center) -- node[midway, above=7mm] {\(c_2\)} (k.east);
     \end{scope}
    \begin{scope}[decoration={brace, raise=6mm, mirror}]
      \draw[decorate] (l.west) -- node[midway, below=7mm] {\(c\)} (k.east);
     \end{scope}
  \end{tikzpicture} }
\end{center}
We let our counting cospan automaton parse the cospan decomposition
$c_1;c_2$ step by step now to show how the annotations for the type
graph $T$ evolve during the process. According to our construction,
every element in $T$ has multiplicity $0$ in the initial state of the
automaton. We then sum up the number of elements within the middle
graphs of the cospans which are \emph{not} part of the right
  interface. Therefore we get the following parsing process:

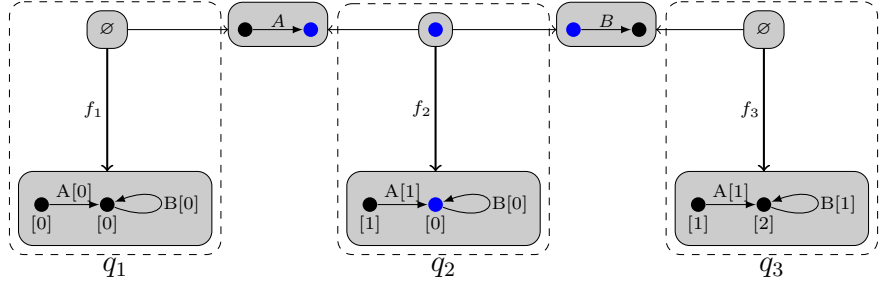
\begin{center}
\scalebox{.97}{
  \begin{tikzpicture}[x=.9cm,y=-1.2cm,baseline=(a1.south)]
    \begin{scope}[shift={(-2.5,0)}]
      \node[glab] (lab1) at (0,0) {$\emptyset$} ;
      \graphboxthick[l]{ (lab1) }  
    \end{scope}
    \begin{scope}
      \node[gnode] (a1) at (-.4,0) {} ;
      \node[gnode,fill=blue] (a2) at (.6,0) {} ;
      \draw[gedge] (a1) -- node[glab,above] (labA1) {$\mathit{A}$} (a2) ;
      \graphboxthick[i]{(a1) (a2) (labA1)}  
    \end{scope}
    \begin{scope}[shift={(2.5,0)}]
      \node[gnode,fill=blue] (a1) at (0,0) {};
      \graphboxthick[r]{ (a1) }  
    \end{scope}
    \begin{scope}[shift={(5,0)}]
      \node[gnode,fill=blue] (a2) at (-.4,0) {} ;
      \node[gnode] (a3) at (.6,0) {} ;
      \draw[gedge] (a2) -- node[glab,above] (labB1) {$\mathit{B}$} (a3) ;
      \graphboxthick[g]{ (a2) (a3) (labB1)}  
    \end{scope}
    \begin{scope}[shift={(7.5,0)}]
      \node[glab] (lab1) at (0,0) {$\emptyset$} ;
      \graphboxthick[k]{ (lab1) }  
    \end{scope}
    \draw[->] \substy{(l.east)}{0} -- \substy{(i.west)}{0} ;
    \draw[->] \substy{(r.west)}{0} -- \substy{(i.east)}{0} ;
    \draw[->] \substy{(r.east)}{0} -- \substy{(g.west)}{0} ;
    \draw[->] \substy{(k.west)}{0} -- \substy{(g.east)}{0} ;
    \begin{scope}[shift={(-2.5,2)}]
      \node[gnode] (a1) at (-1,0) {} ;
      \node[glab,below] (labm1) at (a1.south) {[0]} ;
      \node[gnode] (a2) at (0,0) {} ;
      \node[glab,below] (labm2) at (a2.south) {[0]} ;
      \draw[gedge] (a1) to node[arlab,above] (lab1) {A[0]} (a2) ;
      \draw[gedge] (a2) .. controls +(-20:1cm) and +(20:1cm) .. 
              node[arlab,right] (labm3) {B[0]} (a2) ;
      \graphboxthick[t1]{ (a1) (a2) (labm3) (lab1) (labm1) }  
    \end{scope}
    \begin{scope}[shift={(2.5,2)}]
      \node[gnode] (a1) at (-1,0) {} ;
      \node[glab,below] (labm1) at (a1.south) {[1]} ;
      \node[gnode,fill=blue] (a2) at (0,0) {} ;
      \node[glab,below] (labm2) at (a2.south) {[0]} ;
      \draw[gedge] (a1) to node[arlab,above] (lab1) {A[1]} (a2) ;
      \draw[gedge] (a2) .. controls +(-20:1cm) and +(20:1cm) .. 
              node[arlab,right] (labm3) {B[0]} (a2) ;
      \graphboxthick[t2]{ (a1) (a2) (labm3) (lab1) (labm1) }  
    \end{scope}
    \begin{scope}[shift={(7.5,2)}]
      \node[gnode] (a1) at (-1,0) {} ;
      \node[glab,below] (labm1) at (a1.south) {[1]} ;
      \node[gnode] (a2) at (0,0) {} ;
      \node[glab,below] (labm2) at (a2.south) {[2]} ;
      \draw[gedge] (a1) to node[arlab,above] (lab1) {A[1]} (a2) ;
      \draw[gedge] (a2) .. controls +(-20:1cm) and +(20:1cm) .. 
              node[arlab,right] (labm3) {B[1]} (a2) ;
      \graphboxthick[t3]{ (a1) (a2) (labm3) (lab1) (labm1) }  
    \end{scope}
    \draw[->,thick] \substx{(l.south)}{-2.5} -- node[arlab,left] (f1) 
    {$f_1$} \substx{(t1.north)}{-2.5} ;
    \draw[->,thick] \substx{(r.south)}{2.5} -- node[arlab,left] (f2) 
        {$f_2$} \substx{(t2.north)}{2.5} ;
    \draw[->,thick] \substx{(k.south)}{7.5} -- node[arlab,left] (f3) 
        {$f_3$} \substx{(t3.north)}{7.5} ;
    
    \interfacebox[q1]{(l) (t1)};
    \interfacebox[q2]{(r) (t2)};
    \interfacebox[q3]{(k) (t3)};
    \node[glab,below] (q1lab) at (q1.south) {\large{$q_1$}} ;
    \node[glab,below] (q2lab) at (q2.south) {\large{$q_2$}} ;
    \node[glab,below] (q3lab) at (q3.south) {\large{$q_3$}} ;
  \end{tikzpicture} }
\end{center}

\noindent We visited three states $q_1,q_2$ and $q_3$ in the automaton
with $(q_1,q_2) \in \mathcal{C}(c_1)$ and
$(q_2,q_3) \in \mathcal{C}(c_2)$. Since $\mathcal{C}$ is supposed to
be a functor we get that
$\mathcal{C}(c_1);\mathcal{C}(c_2) = \mathcal{C}(c)$ and therefore
$(q_1,q_3) \in \mathcal{C}(c)$ also holds. In addition we have
$q_1 \in I$ and since the annotation function $b \in \mathcal{B}^2(T)$
in $q_3 = (f_3,b)$ satisfies $\ell \leq b \leq u$ we can infer that
$q_3 \in F$. Therefore we can conclude that
$G \in \mathcal{L}_{\mathcal{C}}$ holds as well.
\end{example}

We still need to prove that $\mathcal{C}$ is indeed a
functor. Intuitively this shows that acceptance of a graph by the
automaton is not dependent on its specific decomposition.


\newcommand{\countingCospanFunctor}{ Let
  $c_1 \colon J \to G \leftarrow K$ and
  $c_2 \colon K \to H \leftarrow L$ be two arrows and let
  $id_G \colon G \to G \leftarrow G$ be the identity cospan. 

  The mapping
  $\mathcal{C}_{T[M]} \colon \emph{Cospan}_m(\GR) \to \mathbf{Rel}$ is
  a functor:
\begin{enumerate}
  \item $\mathcal{C}_{T[M]}(id_G) = id_{C_{T[M]}(G)} $
  \item $\mathcal{C}_{T[M]}(c_1;c_2) = 
  \mathcal{C}_{T[M]}(c_1);\mathcal{C}_{T[M]}(c_2)$
\end{enumerate} }

\begin{proposition}
  \label{prop:c is a functor}
  \countingCospanFunctor
\end{proposition}

\noindent The language accepted by the automaton
$\mathcal{L}_{\mathcal{C}}$ is exactly the graph language
$\mathcal{L}(T[M])$.

\newcommand{\languageEquivalenceCountingCospan}{ Let the multiply
  annotated type graph $T[M]$ (over $\mathcal{B}^n$) and the
  corresponding automaton functor
  $\mathcal{C} \colon \emph{Cospan}_m(\GR) \to \mathbf{Rel}$ for
  $T[M]$ be given. Then
  $\mathcal{L}_{\mathcal{C}} = \mathcal{L}(T[M])$ holds, i.e. for a
  graph $G$ we have $G\in \mathcal{L}(T[M])$ if and only if there
  exist states $i \in I \subseteq \mathcal{C}(\emptyset)$ and
  $f \in F \subseteq \mathcal{C}(\emptyset)$ such that
  $(i,f) \in \mathcal{C}(c)$, where
  $c\colon \emptyset \to G\leftarrow \emptyset$. }

\begin{proposition}
  \label{prop:automaton functor language equivalence}
  \languageEquivalenceCountingCospan
\end{proposition}

Therefore we can construct an automaton for each graph language
specified by a multiply annotated type graph $T[M]$, which accepts
exactly the same language.  In case of a bounded graph language this
automaton will have only finitely many states. Furthermore we can
restrict the label alphabet, i.e., the cospans by using only atomic
cospans, adding a single node or edges (see
\cite{bbek:implementation-graph-automata}). Once these steps are
performed, we obtain conventional non-deterministic finite automata over
a finite alphabet and we can use standard techniques from automata
theory to solve the language inclusion problem directly on the finite
automata.


\subsection{Closure Properties for Multiply Annotated Graphs}
\label{sec:cpannot}

Extending the expressiveness of the type graphs by adding
multiplicities gives us positive results in case of closure under
union and intersection. Here we use constructions that rely on 
products and coproducts in the category of graphs.

\noindent Closure under intersection holds for the most general form of
annotations. From $T_1[M_1]$, $T_2[M_2]$ we can construct an annotated
type graph $(T_1\times T_2)[N]$, where $N$ contains all annotations
which make both projections $\pi_i\colon T_1\times T_2\to T_i$ legal.

\newcommand{\closurePropsAnnotIntersect}{ 
  The category of multiply annotated graphs is closed under 
  intersection.  }

\begin{proposition}
\label{prop:multi annot graphs are closed under intersect}
 \closurePropsAnnotIntersect
\end{proposition}


\noindent We can prove closure under union for the case of annotations based
on the functor $\mathcal{B}^n$. Here we take the coproduct
$(T_1\oplus T_2)[N]$, where $N$ contains all annotations of $M_1$,
$M_2$, transferred to $T_1\oplus T_2$ via the injections
$i_j\colon T_j\to T_1\oplus T_2$. Intuitively, graph items not in
the original domain of the annotations receive annotation $[0,0]$. This can be 
generalized under some mild assumptions%
\full{(see proof in the appendix).}
\short{(see proof in \cite{ckn:specifying-graph-languages-arxiv}).
}


\newcommand{\closurePropsAnnotUnion}{The category of multiply
  annotated graphs over functor $\mathcal{B}^n$ is closed under union.}

\begin{proposition}
  \label{prop:multi annot graphs are closed under union}
  \closurePropsAnnotUnion
\end{proposition}

Closure under complement is still an open issue. If we restrict to
graphs of bounded pathwidth, we have a (non-deterministic) automaton
(functor), as described in Section~\ref{sec:dpannot}, which could be
determinized and complemented. However, this does not provide us with
an annotated type graph for the complement. We conjecture that closure
under complement does not hold.

\section{Conclusion} 
\label{sec:conclusion}

Our results on decidability and closure properties for specification
languages are summarized in the following table. In the case where the
results hold only for bounded pathwidth, the checkmark is in brackets.


\begin{center}
  \begin{tabular}{|c|l||c|c|c|c|}
   \hline
   \multicolumn{2}{|c||}{} & Pure TG & Restr. Gr. & TG Logic & Annotated 
   TG \\ \hline
   & $G \in \mathcal{L}$? & \Yes & \Yes & \Yes & \Yes \\
   Decidability &$\mathcal{L} = \emptyset$? & \Yes & \Yes & \Yes & \Yes \\
   & $\mathcal{L}_1 \subseteq \mathcal{L}_2$? & \Yes & \Yes & \Yes & (\Yes) \\
   \hline
   & $\mathcal{L}_1 \cup \mathcal{L}_2$ & \No & \Yes & \Yes & \Yes
   \\
   Closure Properties & $\mathcal{L}_1 \cap \mathcal{L}_2$ & \Yes &
   \No & \Yes & \Yes \\
   & $\GRL \setminus \mathcal{L}$ & \No & \No & \Yes & ?
   \\
   \hline
  \end{tabular}
\end{center}

One open question that remains is whether language inclusion for
annotated type graphs is decidable if we do not restrict to bounded
treewidth. Similarly, closure under complement is still open.

Furthermore, in order to be able to use these formalisms extensively in
applications, it is necessary to provide a mechanism to compute
weakest preconditions and strongest postconditions. This does not seem
feasible for pure type graphs or the type graph logic. Hence, we are
currently working on characterizing weakest preconditions and
strongest postconditions in the setting of annotated type graphs. This
requires a materialisation construction, similar to
\cite{srw:shape-analysis-3vl}, which we plan to characterize
abstractly, exploiting universal properties in category theory.

Note that our annotations are global, i.e., we count \emph{all} items
that are mapped to a specific item in the type graph. This holds also
for edges, as opposed to UML multiplicities, which are local wrt.\ 
the classes which are related by an
edge (i.e., an association). We plan to study the possibility to
integrate this into our framework and investigate the corresponding
decidability and closure properties.


\noindent \emph{Related work:} As already mentioned there are many
approaches for specifying graph languages. One can not say that one is
superior to the other, usually there is a tradeoff between
expressiveness and decidability properties, 
furthermore they differ in terms of closure properties.

Recognizable graph languages
\cite{c:mso-graphs-I,ce:graph-structure-mso}, which are the
counterpart to regular word languages, 
are closely related with monadic second-order graph logic. If one
restricts recognizable graph languages to bounded treewidth (or
pathwidth as we did), one obtains satisfactory decidability
properties. On the other hand, the size of the resulting graph
automata is often quite intimidating
\cite{bbek:implementation-graph-automata} and hence they are difficult
to work with in practical applications. The use of nested application 
conditions \cite{hp:nested-constraints},
equivalent to first-order logic \cite{r:representing-fol}, has a long
tradition in graph rewriting and they can be used to compute pre- and
postconditions for rules \cite{p:development-correct-gts}. However,
satisfiability and implication are undecidable for first-order logic.

A notion of grammars that is equivalent to context-free (word)
grammars are hyperedge replacement grammars \cite{h:hyper}. Many
aspects of the theory of context-free languages can be transferred
to the graph setting.

In heap analysis the representation of pointer structures to be
analyzed requires methods to specify sets of graphs. Hence both the
TVLA approach by Sagiv, Reps and Wilhelm
\cite{srw:shape-analysis-3vl}, as well as separation logic
\cite{h:resources-concurrency-local,dhy:shape-separation-logic} face
this problem. In \cite{srw:shape-analysis-3vl} heaps are represented
by graphs, annotated with predicates from a three-valued logics (with
truth values \emph{yes}, \emph{no} and \emph{maybe}).

A further interesting approach are forest automata
\cite{ahjltv:heap-forest-automata} that have many interesting
properties, but are somewhat complex to handle.

In \cite{rrlw:diagrammatic-mof-based} the authors study an
approach called Diagram Predicate Framework (DPF), in which type
graphs have annotations based on generalized sketches. This
formalism is intended for MOF-based modelling languages and allows
more complex annotations than our framework.

\bibliography{references}
\bibliographystyle{plain}


\full{
\appendix


\newpage
\section{Proofs}
\label{sec:proofs}

\subsection{Languages Specified by Type or Restriction Graphs}

\begin{proposition_for}{prop:closure properties of pure and restriction}{.}
  \closurePropsPureAndRestrict
\end{proposition_for} 

\begin{proof}
  The product $T_1 \times T_2$ has the property that for any graph $G$
  we have $G \to T_1 \times T_2$ if and only if $G \to T_1$ and
  $G \to T_2$. Hence, given two type graphs $T_1$ and $T_2$, by the
  universal property of the product graph we get immediately the
  following equality:
  $\mathcal{L}(T_1) \cap \mathcal{L}(T_2) = \mathcal{L}(T_1 \times
  T_2)$.

  Dually, given two restriction graphs $R_1$ and $R_2$, we show that
  $\mathcal{L}_R(R_1) \cup \mathcal{L}_R(R_2) =
  \mathcal{L}_R(R_1\oplus R_2)$. In fact,
  $G \not \in \mathcal{L}_R(R_1\oplus R_2)$ iff $R_1\oplus R_2 \to G$,
  iff (by the universal property of coproducts) $R_1 \to G$ and
  $R_2 \to G$, iff $G \not \in \mathcal{L}_R(R_1)$ and
  $G \not \in \mathcal{L}_R(R_2)$, iff
  $G \not \in (\mathcal{L}_R(R_1) \cup \mathcal{L}_R(R_2))$.

For the negative results, we will show counterexamples using the following 
graphs over $\Lambda = 
\{A,B\}$:
  \begin{center}
    \begin{tabular}{ccccc}
      $G_A = $ &
      \begin{tikzpicture}[x=1.5cm,y=-1.2cm,baseline=(1.south)]
      	    \node[gnode] (a1) at (0,0) {} ;
      	    \draw[gedge] (a1) .. controls +(-20:1cm) and +(20:1cm) .. 
      	            node[arlab,right] (lab1) {A} (a1) ;
      \end{tikzpicture}
      & \qquad \qquad \qquad & $G_B = $ &
      \begin{tikzpicture}[x=1.5cm,y=-1.2cm,baseline=(1.south)]
      	    \node[gnode] (a1) at (0,0) {} ;
      	    \draw[gedge] (a1) .. controls +(-20:1cm) and +(20:1cm) .. 
      	            node[arlab,right] (lab1) {B} (a1) ;
      \end{tikzpicture}
    \end{tabular}
  \end{center}

  We first show by contradiction that there is no type graph $T$ such
  that $\mathcal{L}(T) = \mathcal{L}(G_A) \cup \mathcal{L}(G_B)$. In
  fact, the type graph language $\mathcal{L}(G_A)$ contains all graphs
  which do not have any $B$-labeled edge, and $\mathcal{L}(G_B)$
  contains all graphs which do not have any $A$-labeled edge. Since
  $G_A, G_B \in \mathcal{L}(G_A) \cup \mathcal{L}(G_B)$, we would have
  $G_A \to T$ and $G_B \to T$, which implies that $T$ contains at
  least one $A$-labeled loop (thus $T \nrightarrow G_B$) and one
  $B$-labeled loop (thus $T \nrightarrow G_A$). It follows that
  $T \not \in \mathcal{L}(G_A) \cup \mathcal{L}(G_B)$, but instead
  $T \in \mathcal{L}(T)$ yielding a contradiction.

  Now we show by contradiction that there is no restriction graph $R$
  such that
  $\mathcal{L}_R(R) = \mathcal{L}_R(G_A) \cap \mathcal{L}_R(G_B)$. In
  fact, if such an $R$ exists we would have
  $\mathcal{L}_R(R) \subseteq \mathcal{L}_R(G_A)$, and thus
  $R \to G_A$ by Proposition~\ref{prop:decidability problems of pure
    and restriction}.4, and
  $\mathcal{L}_R(R) \subseteq \mathcal{L}_R(G_B)$, and thus
  $R \to G_B$. But $R \to G_A$ means that $R$ has no $B$-edges, and
  $R \to G_B$ that it has no $A$-edges, thus $R$ must be discrete.
  This yields a contradiction, because any graph $G$ with no loops but
  with at least one edge belongs to
  $\mathcal{L}_R(G_A) \cap \mathcal{L}_R(G_B)$ but not to
  $\mathcal{L}_R(R)$, because $R \to G$.  


  The lack of closure under complementation immediately follows from
  these negative results and the fact that union can be expressed
  using intersection and complementation, and dually. \qed
\end{proof}

\begin{proposition_for}{prop:decidability problems of pure and restriction}{.}
  \decidabilityPureAndRestrict
\end{proposition_for} 

\begin{proof}

  \begin{enumerate}
  \item To decide whether $G \in \mathcal{L}(T)$ (or
    $G \in \mathcal{L}_R(R)$) holds, we need to check for the
    existence of a morphism $\varphi \colon G \to T$ (or for the
    non-existence of a morphism $\varphi \colon R \to G$), which is
    obviously possible because graphs are finite.  Nevertheless, note
    that this problem is NP-complete. For instance, searching for a
    morphism from any graph into the 3-clique is the same as deciding
    if the graph is 3-colorable.
  
  \item The emptiness problem is pretty trivial. If
    $\mathcal{L} = \mathcal{L}(T)$ for a type graph $T$, then
    $\mathcal{L}(T) \neq \emptyset$ because it holds
    $\emptyset \in \mathcal{L}(T)$ (recall that $\emptyset$ is the
    initial object of $\GR$).

    If instead $\mathcal{L} = \mathcal{L}_R(R)$ for a restriction
    graph $R$, then $\mathcal{L} = \emptyset$ if and only if
    $R = \emptyset$.  In fact, if $R = \emptyset$ then $R \to G$ for
    all $G \in \GR$, and thus $\mathcal{L}_R(R) = \emptyset$.  Instead
    if $R \not = \emptyset$ then clearly $R \nrightarrow \emptyset$,
    thus $\emptyset \in \mathcal{L}_R(R) \not = \emptyset$.
  
  \item We show that $\mathcal{L}(T_1) \subseteq \mathcal{L}(T_2)$ iff
    $T_1 \to
    T_2$, which is decidable. \\
    $\Rightarrow$: Assume
    $\mathcal{L}(T_1) \subseteq \mathcal{L}(T_2)$ holds.  Since
    $T_1 \in \mathcal{L}(T_1)$ holds then $T_1 \in
    \mathcal{L}(T_2)$ also holds and therefore  $T_1 \to T_2$. \\
    $\Leftarrow$: Assume $T_1 \to T_2$ holds, and let
    $G \in \mathcal{L}(T_1)$.  Therefore $G \to T_1$, and by
    transitivity $G \to T_2$, thus $G \in \mathcal{L}(T_2)$.
  
  \item We show that $\mathcal{L}_R(R_1) \subseteq \mathcal{L}_R(R_2)$
    iff
    $R_1 \to R_2$. \\
    $\Rightarrow$: Assume that
    $\mathcal{L}_R(R_1) \subseteq \mathcal{L}_R(R_2)$
    holds. Equivalently,
    $\{G \mid R_2 \to G\} = \overline{\mathcal{L}_R(R_2)} \subseteq
    \overline{\mathcal{L}_R(R_1)} = \{G \mid R_1 \to G\}$, where we
    wrote $\overline{\mathcal{L}}$ for the complement language
    $(\GRL\setminus
    \mathcal{L})$. Thus, since obviously $R_2 \to R_2$, we obtain $R_1 \to R_2$.\\
    $\Leftarrow$: Assume that $R_1 \to R_2$ holds and that
    $G \in {\mathcal{L}_R(R_1)}$, which means $R_1 \nrightarrow
    G$. If, by contradiction, $G \not \in {\mathcal{L}_R(R_2)}$, then
    we have $R_2 \to G$ and, by transitivity, $R_1 \to G$, which is a
    contradiction.\qed
   \end{enumerate} 
\end{proof}

\begin{proposition_for}{prop:closure-rg}{.}
  \closureUnderRewritingRestrict
\end{proposition_for}

\begin{proof}~
  \noindent $\Leftarrow$: Assume that the condition holds. Now let
  $G,H$ with $G\Rightarrow_\rho H$. Instead of showing that
  $G\in \mathcal{L}_R(S)$ implies $H \in \mathcal{L}_R(S)$, we show
  that $H \not\in \mathcal{L}_R(S)$ implies
  $G\not\in \mathcal{L}_R(S)$.

  Since $G\Rightarrow_\rho H$ we have the following DPO
  diagram (below, on the left) for a rule
  $\rho = (L \arleft[\varphi_L] I \arright[\varphi_R] R) \in
  \mathcal{R}$. Furthermore, since $H \not\in \mathcal{L}_R(S)$, there
  exists a morphism $S\to H$.
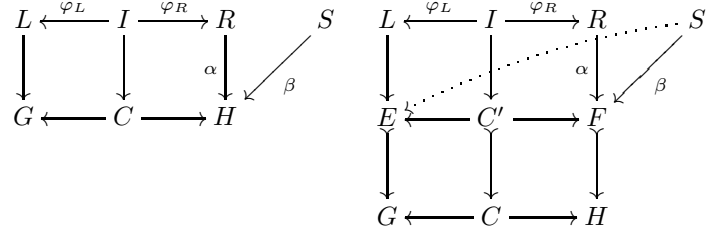
  \[
    \xymatrix{
      L \ar[d] & I \ar[l]_{\phi_L} \ar[r]^{\phi_R} \ar[d] & R
      \ar[d]_{\alpha} 
      & S \ar[ld]^{\beta} \\ 
      G & C \ar[l] \ar[r] & H
    }
    \quad
    \xymatrix{
      L \ar[d] & I \ar[l]_{\phi_L} \ar[r]^{\phi_R} \ar[d] & R
      \ar[d]_{\alpha} 
      & S \ar[ld]^{\beta} \ar@/_.5pc/@{.>}[llld] \\ 
      E \ar@{>->}[d] & C' \ar@{>->}[d] \ar[l] \ar[r] & F \ar@{>->}[d] & \\
      G & C \ar[l] \ar[r] & H &
    }
  \]
  Now take the joint image $F$ of $R$ and $S$ in $H$, i.e., factor the
  morphisms $\alpha$, $\beta$ into $R\to F\rightarrowtail H$ and
  $S\to F\rightarrowtail H$, where the arrows $R\to F$, $S\to F$ are
  jointly epi. Since we are working in an adhesive category, the
  pushouts split into pushouts according to \cite{ls:adhesive-journal}
  (see diagram above, on the right). Now, $E$ is obtained from $F$ by
  applying rule $\rho$ backwards. Hence, the condition implies that
  there exists a morphism $S\to E$ and this means that there is a
  morphism $S\to G$, which implies $G\not\in \mathcal{L}_R(S)$.

  \smallskip

  \noindent $\Rightarrow$: Assume that $\mathcal{L}_R(S)$ is closed
  under rewriting via a rule $\rho$. We show that the condition
  holds. Let $\alpha\colon R\to F$, $\beta\colon S\to F$ be a pair of
  morphisms which are jointly epi and assume that $E$ is obtained from
  $F$ by applying $\rho$ backwards.

  Now, since $E\Rightarrow_\rho F$ and $H\not\in \mathcal{L}_R(S)$, we
  infer that $G\not\in \mathcal{L}_R(S)$, otherwise we would have a
  counterexample to closure under rewriting. Hence there exists a
  morphism $S\to E$. \qed
\end{proof}

For the next result we need to recall the following lemma presented in~\cite{nt:duality-finite-structures}.

\begin{lemma}[Lemma 2.1 of~\cite{nt:duality-finite-structures}]
\label{le:cores}
Let $T$ be a graph and $\core{T}$ be its core. Then for each morphism
$f: T \to \core{T}$ there exists a morphism $f': \core{T} \to T$ such
that $f \circ f' = id_{\core{T}}$. Vice versa, for each morphism
$g': \core{T} \to T$ there exists a morphism $g: T \to \core{T}$ such
that $g \circ g' = id_{\core{T}}$.
\end{lemma}

\begin{proposition_for}{prop:closure-tg}{.}
  \closureUnderRewritingTypeGraphs
\end{proposition_for}

\begin{proof}~

  \noindent $\Rightarrow$: \quad Notice that if
  $L \not \in \mathcal{L}(T)$ then there is no morphism
  $t_L\colon L\to \core{T}$ and we are done. Otherwise, let
  $t_L\colon L\to \core{T}$, and let $n: I \to \core{T}$ be defined as
  $n = t_L \circ \phi_L$.  Consider the following diagram, where the
  top span is the rule, and the two squares are built as pushouts ($A$
  is the pushout of $\phi_L,n$; $B$ is the pushout of $\phi_R,n$).

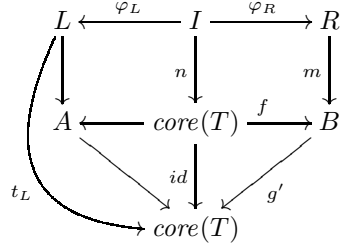
  \centerline{\xymatrix{
      L \ar[d] \ar@/_40pt/[ddr]_{t_L} &
      I \ar[l]_{\phi_L} \ar[r]^{\phi_R} \ar[d]_n&
      R \ar[d]_m \\
      {A } \ar[dr] &
      \core{T} \ar[r]^f  \ar[l] \ar[d]_{id} &
      B \ar[dl]^{g'}\\
      & \core{T} 
    }
  }
  
  Arrow $A \to \core{T}$ is uniquely determined because the left
  square is a pushout and $id \circ n = t_L \circ \phi_L$.  This arrow
  witnesses that $A \in \mathcal{L}(T)$ (because $\core{T} \to T$),
  and thus by assumption $B \in \mathcal{L}(T)$, because obviously
  $A \Rightarrow_\rho B $. Therefore we know that $B \to T$, and thus
  that there is an arrow $g: B \to \core{T}$.  In general, this arrow
  does not make the lower right triangle commute, but given that we
  also have arrow $f: \core{T} \to B$ as the base of the right
  pushout, it follows that $B\sim T$ and hence
  $\core{B} \cong \core{T}$. Therefore by Lemma~\ref{le:cores}, we
  know that there is an arrow $g': B \to \core{T}$ such that
  $g' \circ f = id_{\core{T}}$ (in particular,
  $g' = (g\circ f)^{-1}\circ g$). Therefore in the above diagram also
  the lower right triangle commutes, and arrow
  $t_R = g' \circ m: R \to \core{T}$ satisfies
  $t_L\circ\varphi_L = t_R\circ\varphi_R$, as desired.

  \noindent $\Leftarrow$: \quad Assume that $G\in \mathcal{L}(T)$ by
  morphism $t_G: G \to T$, and that $G$ is rewritten to $H$ via rule
  $\rho = (L \arleft[\varphi_L] I \arright[\varphi_R] R)\in
  \mathcal{R}$. Also, let $t$ be any arrow from $T$ to $\core{T}$.
  This gives us the diagram below, where the two squares are pushouts,
  and the left triangle commutes by taking for $C \to \core{T}$ the
  composition $t \circ t_G\circ\psi_L$.
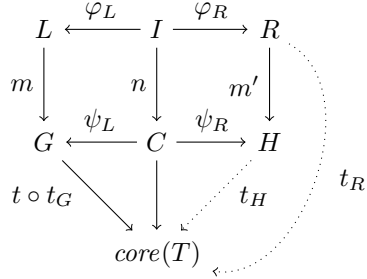
  \begin{center}
    \begin{tikzpicture}[shorten >=1pt, node distance=15mm and 15mm, on grid]
      \draw
      node (L) at (0,0) {\(L\)}
      node [right=of L] (I) {\(I\)}
      node [right=of I] (R) {\(R\)}
      node [below=of L] (G) {\(G\)}
      node [below=of I] (C) {\(C\)}
      node [below=of R] (H) {\(H\)}
      node [below=of C] (T) {\(\core{T}\)};
      \begin{scope}[->]
        \draw (L) -- node[midway, left=.1mm] {\(m\)} (G);
        \draw (I) -- node[midway, above=.1mm] {\(\varphi_L\)} (L);
        \draw (I) -- node[midway, left=.1mm] {\(n\)} (C);
        \draw (I) -- node[midway, above=.1mm] {\(\varphi_R\)}(R);
        \draw (R) -- node[midway, left=.1mm] {\(m'\)} (H);
        \draw (C) -- node[midway, above=.1mm] {\(\psi_L\)} (G);
        \draw (C) -- node[midway, above=.1mm] {\(\psi_R\)} (H);
        \draw (G) -- node[midway, left=2mm] {\(t \circ t_G \)} (T);
        \draw (C) -- (T);
        \draw [dotted] (R) edge[bend left=80] node[midway, right=2mm] {\(t_R\)} 
        (T);
        \draw [dotted] (H) -- node[midway, right=2mm] {\(t_H\)} (T);
      \end{scope}
    \end{tikzpicture}
  \end{center}

  \noindent
  By assumption, since 
  $t\circ t_G\circ m\colon L\to \core{T}$, there exists a morphism
  $t_R\colon R\to \core{T}$ such that
  $t\circ t_G\circ m\circ \phi_L = t_R\circ \phi_R$. This means that the
  square consisting of $I,C,R,\core{T}$ commutes, that is
  $t \circ t_G\circ\psi_L\circ n = t \circ t_G\circ m\circ \phi_L =
  t_R\circ\phi_R$. Hence there exists a mediating morphism
  $t_H\colon H\to \core{T}$, which implies $H\in \mathcal{L}(T)$ because $\core{T} \to T$. \qed
\end{proof}

\subsection{Type Graph Logic}

\begin{proposition_for}{prop:decidability problems of the type graph logic}{.}
  \decidabilityTypeGraphLogic
\end{proposition_for}

\begin{proof}~
  
  \noindent \emph{Membership:} The membership problem for graph
  languages over $\TGL$ formulae is decidable since it is decidable
  for every type graph language $\mathcal{L}(T)$. We simply build the
  syntax tree of the formula $F$ and search for morphisms
  $\varphi_i \colon G \to T_i$ at the leafs of the tree.  Afterwards
  we pass the boolean results up to the root to decide whether
  $G \in \mathcal{L}(F)$ holds.
  
  \noindent \emph{Emptiness:} In order to show whether
  $\mathcal{L}(F) = \emptyset$ holds, we transform $F$ into
  disjunctive normal form (DNF). 
  It is sufficient to check whether all conjunctions of the form
  $(T_0\land \lnot T_1\land \dots \land \lnot T_n)$ are unsatisfiable.
  We can assume that there is at most one positive type graph in every
  conjunction, since type graphs are closed under
  conjunction/intersection. Furthermore we can even assume that there
  is exactly one positive type graph, since we always add
  $T_{\text{\ding{92}}}$ (the flower graph).

  Now we have: 
  \begin{eqnarray*}
    && \mathcal{L}(T_0 \land \lnot T_1
    \land \ldots \land \lnot T_n) = \emptyset \\
    & \iff & \mathcal{L}(T_0 \land \lnot (T_1
    \lor \ldots \lor T_n) ) = \emptyset \\
    & \iff & \mathcal{L}(T_0) \cap \overline{\mathcal{L}
    (T_1 \lor \ldots \lor T_n)} = \emptyset \\
    & \iff & \mathcal{L}(T_0) \subseteq
    \mathcal{L}(T_1 \lor \ldots \lor T_n) \\
    & \iff & \mathcal{L}(T_0) \subseteq
    \mathcal{L}(T_1) \cup \ldots \cup \mathcal{L}(T_n) \\
    & \iff & \exists \varphi \colon T_0 \to T_k \quad \mbox{for some index
      $1\le k\le n$}
  \end{eqnarray*}

  \noindent
  Therefore, we need to check whether for each of the conjunctions
  $(T_0\land \lnot T_1\land \dots \land \lnot T_n)$ in the DNF of $F$,
  there exists a morphism $\varphi \colon T_0 \to T_k$ for some
  $1 \leq k \leq n$ .
  
  \noindent \emph{Inclusion:} The language inclusion problem can be
  reduced to the aforementioned emptiness problem. To solve the
  language inclusion we use the following equivalence:
  $$ \mathcal{L}(F_1) \subseteq \mathcal{L}(F_2) \iff \mathcal{L}(F_1 \land 
  \lnot F_2) = \emptyset $$ 
  Since the emptiness problem is decidable we can conclude that the language 
  inclusion problem is decidable as well. \qed
\end{proof}



\subsection{Annotated Type Graphs}

\begin{lemma_for}{le:composition of legal morphism}{.}
  \compOfLegalMorph
\end{lemma_for}

\begin{proof}
  Let $\varphi_1 \colon T_1[M_1] \to T_2[M_2]$ and
  $\varphi_2 \colon T_2[M_2] \to T_3[M_3]$ be two legal morphisms in
  the category of multiply annotated graphs. Since $\varphi_1$ is
  legal we get that for all $(\ell_1,u_1) \in M_1$ there exists
  $(\ell_2,u_2) \in M_2$ such that
  $\ell_2 \leq \mathcal{A}_{\phi_1}(\ell_1)$ and
  $\mathcal{A}_{\phi_1}(u_1) \leq u_2$ hold. Furthermore $\varphi_2$
  is legal and therefore we get that for all $(\ell_2,u_2) \in M_2$
  there exists $(\ell_3,u_3) \in M_3$ such that
  $\ell_3 \leq \mathcal{A}_{\phi_2}(\ell_2)$ and
  $\mathcal{A}_{\phi_2}(u_2) \leq u_3$ hold as well. We define
  $\varphi \colon T_1[M_1] \to T_3[M_3]$ to be the composed morphism
  with $\varphi = \varphi_2 \circ \varphi_1$ and due to the fact that
  $\mathcal{A}$ is a functor which preserves monotonicity, we get the
  following two inequalities:

  \vspace*{-.5cm}\begin{center}
  \begin{minipage}[b]{.45\linewidth}
    \begin{align*}
        \phantom{\equiv}&& \ell_2 &\leq \mathcal{A}_{\phi_1}(\ell_1) \\
        \Rightarrow && \mathcal{A}_{\phi_2}(\ell_2) &\leq 
        \mathcal{A}_{\phi_2}(\mathcal{A}_{\phi_1}(\ell_1)) \\
        \Rightarrow && \ell_3 \leq \mathcal{A}_{\phi_2}(\ell_2) &\leq 
        \mathcal{A}_{\phi_2 \circ \phi_1}(\ell_1) \\
        \Rightarrow && \ell_3 &\leq 
        \mathcal{A}_{\phi}(\ell_1) 
      \end{align*}
  \end{minipage}
  \rule{1pt}{2.1cm} \ 
  \begin{minipage}[b]{.48\linewidth}
    \begin{align*}
        \phantom{\equiv}&&
        \mathcal{A}_{\phi_1}(u_1) &\leq u_2 \\
        \Rightarrow && 
        \mathcal{A}_{\phi_2}(\mathcal{A}_{\phi_1}(u_1)) &\leq 
        \mathcal{A}_{\phi_2}(u_2)  \\
        \Rightarrow && 
        \mathcal{A}_{\phi_2 \circ \phi 1}(u_1) &\leq 
        \mathcal{A}_{\phi_2}(u_2) \leq u_3 \\
        \Rightarrow && 
        \mathcal{A}_{\phi}(u_1) &\leq u_3
      \end{align*}
  \end{minipage}
  \end{center}

  \noindent Since both $\ell_3 \leq \mathcal{A}_{\phi}(\ell_1)$ and
  $\mathcal{A}_{\phi}(u_1) \leq u_3$ hold, the morphism $\varphi$ is
  legal. \qed
\end{proof}

\begin{proposition_for}{prop:emptiness is decidable for annotated type 
graphs}{.}
  \emptinessCheckAnnot
\end{proposition_for}

\begin{proof}
  \noindent $\Leftarrow$: Assume that $M=\emptyset$, in this case
  $\mathcal{L}(T[M])$ is clearly empty as well. Assume that there is
  an annotation $(\ell,u)\in M$ such that $\ell(e) \ge 1$ and
  ($u(\sSrc(e)) = 0$ or $u(\sTgt(e)) = 0$) for some edge $e\in E_T$. Then
  no graph can satisfy these lower and upper bounds, since we are
  forced to map at least one edge to $e$, but are not allowed to map
  any node to the source respectively target node. If this is true for
  all annotations, the language of the type graph must be empty.

  \noindent $\Rightarrow$: Now assume that $M \neq \emptyset$ and
  there exists one annotation $(\ell,u)\in M$ such that for every
  edge $e\in E_T$ with $\ell(e) \ge 1$ we have $u(\sSrc(e)) \ge 1$ and
  $u(\sTgt(e)) \ge 1$). 

  Now take $T[\ell,u]$ and remove from $T$ all edges and nodes $x$
  with $u(x) = 0$, resulting in a graph $T'$. If a node is removed
  all incident edges are removed as well. Note that in such a case
  only edges $e$ with $\ell(e) = 0$ will be removed (due to the
  condition above).

  Now define $\ell' = \ell|_{T'}$ and $u' = u|_{T'}$. Due to the considerations 
  above there exists a legal morphisms (embedding)
  $T'[\ell',u']\mor T[\ell,u]$, since the removed items
  had a lower bound of $0$.  Furthermore each remaining item has an
  upper bound of at least $1$, i.e., it represents at least one node
  or edge.

  Now construct a graph $G$ from $T'$ by proceeding as follows: for
  every node $v$ with $\ell'(v) = k$ add $k-1$ isolated nodes (zero
  isolated nodes if $k=0$). For every edge $e$ with $\ell'(e) = k$ put
  $k$ parallel edges between $\sSrc(e),\sTgt(e)$. There is a morphism
  $\phi\colon G\to T'$ obtained by mapping every item to the item from
  which it orginated.

  Mapping $G[s_G,s_G]$ to $T'$ via $\phi$ will give us an annotation
  $\mathcal{B}^n_\phi(s_G)$. This annotation will coincide will the
  lower bound $\ell$ in all cases, apart from the case where there is a
  node $v$ with $\ell(v) = 1$. In this case
  $\mathcal{B}^n_\phi(s_G)(v) = 1$, but this is covered by the upper
  bound which is at least $1$.

  Hence there is a legal graph morphism from $G[s_G,s_G]$ to
  $T'[\ell',u']$ and -- by composition -- to $T[\ell,u]$. Hence
  $G\in \mathcal{L}(T[M])$ and hence $\mathcal{L}(T[M])\neq \emptyset$.
  %
  %
  \qed
\end{proof}

\begin{proposition_for}{prop:language inclusion using legal morphisms}{.}
  \languageinclusionWithLegalMorphism
\end{proposition_for}

\begin{proof}
  Every graph $G \in \mathcal{L}(T_1[M])$ has a legal morphism
  $\varphi' \colon G[s_G,s_G] \to T_1[M]$. Whenever there exists a
  legal morphism $\varphi \colon T_1[M] \to T_2[N]$ between the two
  multiply annotated type graphs, we obtain the morphism
  $\eta \colon G[s_G,s_G] \to T_2[N]$ with
  $\eta = \varphi \circ \varphi'$ which is legal due to Lemma
  \ref{le:composition of legal morphism}. Therefore $G \in T_2[N]$
  holds and we can conclude that
  $\mathcal{L}(T_1[M]) \subseteq \mathcal{L}(T_2[N])$ also holds. \qed
\end{proof}

\begin{proposition_for}{prop:c is a functor}{.}
  \countingCospanFunctor
\end{proposition_for}

\begin{proof} 

  \noindent $1.$ The identity relation $id_{C_{T[M]}(G)}$ consists of
  all pairs $(i,i)$ with $i \in C_{T[M]}(G)$. Let the two states
  $i,j \in \mathcal{C}_{T[M]}(G)$ be given with
  $i = (f_1 \colon G \to T,b_1)$ and $j = (f_2 \colon G \to
  T,b_2)$. The pair $(i,j)$ is in the relation
  $\mathcal{C}_{T[M]}(id_G)$ if and only if there exists a morphism
  $h \colon G \to T$ such that, for all\footnote{We write $x\in T$ as
    an abbreviation for $x\in V_T\cup E_T$.} $x \in T$ the equation
  $b_2(x) = b_1(x) + | \{ y \in (G \setminus id(G)) \mid h(y)=x \}|$
  holds and the following diagram commutes:
  \begin{center}
    \begin{tikzpicture}[shorten >=1pt, node distance=15mm and 15mm, on grid]
      	\draw
      		node (J) at (0,0) {\(G\)}
      		node [right=of J] (G) {\(G\)}
      		node [right=of G] (K) {\(G\)}
      		node [below=of G] (T) {\(T\)};
    
      	\begin{scope}[decoration={brace, raise=6mm}]
      		\draw[decorate] (J.west) -- node[midway, above=6mm] {\(id_G \colon G 
      		\mor G\)} (K.east);
    
      	\end{scope}	
      	\begin{scope}[->]
      		\draw (J) -- node[midway, left=2mm] {\(f_1\)}  (T);
      		\draw (J) -- node[midway, above=.1mm] {\(id\)} (G);
      		\draw[dashed] (G) -- node[above right=.2mm] {\(\exists h\)} (T);
      		\draw (K) -- node[midway, above=.1mm] {\(id\)}(G);
      		\draw (K) -- node[midway, right=2mm] {\(f_2\)} (T);
      	\end{scope}
    \end{tikzpicture}
  \end{center}
Since the diagram commutes we obtain that $f_1 = f_2$ since $f_1 = id;h = h;id 
= 
f_2$ holds and for all $x \in T$ the annotation functions $b_1$ and $b_2$ are 
equal due to the following equation:
\begin{align*}
  b_2(x) &= b_1(x) + | \{ y \in (G \setminus id(G)) \mid h(y)=x \}| \\
         &= b_1(x) + | \{ y \in \emptyset \mid h(y)=x \}| = b_1(x) + 0 = b_1(x)
\end{align*}
This is equivalent to $i=j$ and therefore for all $i \in C_{T[M]}(G)$ the 
following equation holds:
\begin{align*}
\mathcal{C}_{T[M]}(id_G) &= \{ (i,j) \in C_{T[M]}(G) \times C_{T[M]}(G) \mid 
i=j \} \\
&= id_{C_{T[M]}(G)}
\end{align*}
Therefore $\mathcal{C}_{T[M]}(id_G) = id_{C_{T[M]}(G)}$ holds. \\

\noindent In the following part let $c_1 \colon J \arright[g_1] G \arleft[g_2] 
K$ and $c_2 \colon K \arright[g_1'] H \arleft[g_2'] L$ be given and let $c = 
c_1;c_2$ with $c \colon J \arright[g_1;j_1] G' \arleft[g_2';j_2] L$ be the 
composed morphism of $c_1$ and $c_2$. \\

\noindent $2. \subseteq :$  Let $(i,j) \in 
\mathcal{C}_{T[M]}(c_1;c_2)$ be given with $i \in \mathcal{C}_{T[M]}(J)$ 
and $j \in \mathcal{C}_{T[M]}(L)$ such that $i = (f_1 \colon J \to T, 
b_1)$ and $j = (f_3 \colon L \to T, b_3)$. Then there exists a 
morphism $h \colon G' \to T$ such that $b_3(x) = b_1(x) + | \{ y 
\in (G'\setminus g_2';j_2(L)) \mid h(y)=x \}|$ holds for all $x \in T$ and the 
following diagram commutes:

  
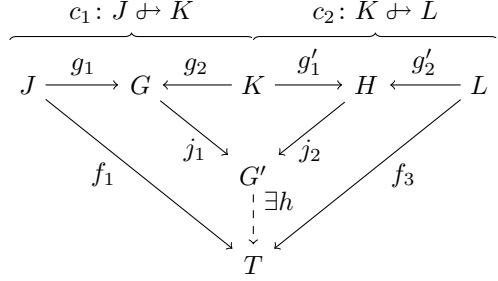
\begin{center}
    \begin{tikzpicture}[shorten >=1pt, node distance=12mm and 15mm, on grid]
      	\draw
      		node (J) at (0,0) {\(J\)}
      		node [right=of J] (G) {\(G\)}
      		node [right=of G] (K) {\(K\)}
      		node [right=of K] (H) {\(H\)}
      		node [right=of H] (L) {\(L\)}
      		node [below=of K] (G') {\(G'\)}
      		node [below=of G'] (T) {\(T\)};
    
      	\begin{scope}[decoration={brace, raise=6mm}]
      		\draw[decorate] (J.west) -- node[midway, above=7mm] {\(c_1 \colon J 
      		\mor K\)} (K.center);
      	\end{scope}	
      	\begin{scope}[decoration={brace, raise=6mm}]
      		\draw[decorate] (K.center) -- node[midway, above=7mm] {\(c_2 \colon K 
      		\mor L\)} (L.east);
      	\end{scope}	

      	\begin{scope}[->]
      		\draw (J) -- node[midway, left=2mm] {\(f_1\)}  (T);
      		\draw (J) -- node[midway, above=.1mm] {\(g_1\)} (G);
      		\draw (K) -- node[midway, above=.1mm] {\(g_1'\)} (H);
      		\draw (L) -- node[midway, above=.1mm] {\(g_2'\)} (H);
      		\draw (G) -- node[midway, below=.1mm] {\(j_1\)} (G');
      		\draw (H) -- node[midway, below=.1mm] {\(j_2\)} (G');
      		\draw[dashed] (G') -- node[above right=.2mm] {\(\exists h\)} (T);
      		\draw (K) -- node[midway, above=.1mm] {\(g_2\)}(G);
      		\draw (L) -- node[midway, right=2mm] {\(f_3\)} (T);
      	\end{scope}
    \end{tikzpicture}
  \end{center}

\noindent To prove that $(i,j) \in 
\mathcal{C}_{T[M]}(c_1);\mathcal{C}_{T[M]}(c_2)$ is satisfied from the above 
properties, we need to show that there 
exists a $k \in \mathcal{C}_{T[M]}(K)$ where $k = (f_2 \colon K \to T, b_2)$ 
such that $(i,k) \in \mathcal{C}_{T[M]}(c_1)$ 
 and $(k,j) \in \mathcal{C}_{T[M]}(c_2)$. 
Let $f_2 = g_2;j_1;h = g_1';j_2;h$. Then there must exist 
two morphisms $h_1 \colon G \to T$, $h_2 \colon H \to T$ such that the 
following six properties hold:
$$
(i,k) \in \mathcal{C}_{T[M]}(c_1)
\left\{
	\begin{array}{ll}
		(1)  & g_1;h_1 = f_1 \\
		(2)  & g_2;h_1 = f_2 \\
		(3)  & \forall x \in T \quad b_2(x) = b_1(x) + | \{ y \in (G \setminus 
		g_2(K)) \mid h_1(y)=x \}|
	\end{array}
\right. $$
$$
(k,j) \in \mathcal{C}_{T[M]}(c_2)
\left\{
	\begin{array}{ll}
		(4)  & g_1';h_2 = f_2 \\
		(5)  & g_2';h_2 = f_3 \\
		(6)  & \forall x \in T \quad b_3(x) = b_2(x) + | \{ y \in (H \setminus 
		g_2'(L)) \mid h_2(y)=x \}|
	\end{array}
\right.
$$ 
We define $h_1,h_2$ to be $h_1 = j_1;h$ and $h_2 = j_2;h$ which
already satisfy the following four properties:
\begin{align*}
(1) \quad & g_1;h_1 = g_1;j_1;h = f_1 && (4) \quad g_1';h_1 = g_1';j_2;h = 
f_2 \\
(2) \quad & g_2;h_1 = g_2;j_1;h = f_2 && (5) \quad g_2';h_1 = g_2';j_2;h = 
f_3 
\end{align*}
We define $b_2$ with respect to property $(3)$, such that for all $x \in T$  
the equation $b_2(x) = b_1(x) + | \{ y \in (G \setminus g_2(K)) \mid h_1(y)=x 
\}|$ holds. Due to property $(6)$ we get the following equation for $b_3$, 
for which we need to prove that it holds for all elements $x \in T$:
$$ b_3(x) = b_1(x) + | \{ y \in (G \setminus g_2(K)) \mid h_1(y)=x 
\}| + | \{ y \in (H \setminus g_2'(L)) \mid h_2(y)=x \}| $$ We prove
the following equation $(7)$ instead for all $x \in T$, from which we
can easily derive afterwards that also $(6)$ holds: 
\begin{align*}
  & | \{ y \in (G \setminus g_2(K)) \mid h_1(y)=x 
  \}| + | \{ y \in (H \setminus g_2'(L)) \mid h_2(y)=x \}| \\
  = &| \{ y \in (G'\setminus 
  g_2';j_2(L)) \mid h(y)=x \}| \tag{7}
\end{align*}
Since the morphisms $j_1$ and $j_2$ are both injective and $G'$ is the
pushout object of $G$ and $H$ over the common graph $K$ we get that
$G' = j_1(G \setminus g_2(K)) \uplus j_2(H)$. Subtracting all elements
$x \in L$ that are being mapped into $H$ on both sides of the
equation, we get that
$G'\setminus g_2';j_2(L) = j_1(G \setminus g_2(K)) \uplus j_2(H
\setminus g_2'(L))$ holds as well. Using this fact we can prove
equation $(7)$ which holds for all $x \in T$:
\begin{align*}
  &| \{ y \in (G'\setminus g_2';j_2(L)) \mid h(y)=x \}| \\
  = &| \{ y \in \bigl(j_1(G \setminus g_2(K)) \uplus j_2(H \setminus 
  g_2'(L))\bigr) \mid 
  h(y)=x \}| \\
  = &| \{ y \in j_1(G \setminus g_2(K)) \mid 
    h(y)=x \} \uplus \{ y \in j_2(H \setminus 
      g_2'(L)) \mid h(y)=x \}  | \\
  = &| \{ y \in G \setminus g_2(K) \mid 
    (j_1;h)(y)=x \} \uplus \{ y \in H \setminus 
      g_2'(L) \mid (j_2;h)(y)=x \}  | \\
  = &| \{ y \in G \setminus g_2(K) \mid 
    h_1(y)=x \} \uplus \{ y \in H \setminus 
      g_2'(L) \mid h_2(y)=x \}  | \\
  = & | \{ y \in (G \setminus g_2(K)) \mid h_1(y)=x 
  \}| + | \{ y \in (H \setminus g_2'(L)) \mid h_2(y)=x \}| 
\end{align*}
Using equation $(7)$ we conclude that property $(6)$ always holds for all $x 
\in T$:
\begin{align*}
b_3(x) &= b_1(x) + | \{ y \in (G'\setminus g_2';j_2(L)) \mid h(y)=x \}| \\
&= b_1(x) + | \{ y \in (G \setminus g_2(K)) \mid h_1(y)=x 
\}| + | \{ y \in (H \setminus g_2'(L)) \mid h_2(y)=x \}| \\
&= b_2(x) + | \{ y \in (H \setminus g_2'(L)) \mid h_2(y)=x \}| 
\end{align*}
Therefore $(i,j) \in \mathcal{C}_{T[M]}(c_1);\mathcal{C}_{T[M]}(c_2)$ holds as 
well. \\

\noindent $2. "\supseteq" :$ Let two pairs $(i,k) \in \mathcal{C}_{T[M]}(c_1)$ 
and $(k,j) \in \mathcal{C}_{T[M]}(c_2)$ be given with $i \in 
\mathcal{C}_{T[M]}(J)$, 
$k \in \mathcal{C}_{T[M]}(K)$ and $j \in \mathcal{C}_{T[M]}(L)$ 
such that $i = (f_1 \colon J \to T, b_1)$, $k = (f_2 \colon K \to T, b_2)$ and 
$j = (f_3 \colon L \to T, b_3)$. Then in addition there exist two morphisms 
$h_1 \colon G \to T$ and $h_2 \colon H \to T$ such that for all $x \in T$ the 
two equations $b_2(x) = b_1(x) + | \{ y \in (G \setminus g_2(K)) \mid h_1(y)=x 
\}|$ and \\ $b_3(x) = b_2(x) + | \{ y \in (H \setminus 
		g_2'(L)) \mid h_2(y)=x \}|$ both hold and the following diagram 
		commutes:

  \begin{center}
    \begin{tikzpicture}[shorten >=1pt, node distance=12mm and 15mm, on grid]
      	\draw
      		node (J) at (0,0) {\(J\)}
      		node [right=of J] (G) {\(G\)}
      		node [right=of G] (K) {\(K\)}
      		node [right=of K] (H) {\(H\)}
      		node [right=of H] (L) {\(L\)}
      		node [below=of K] (G') {}
      		node [below=of G'] (T) {\(T\)};
    
      	\begin{scope}[decoration={brace, raise=6mm}]
      		\draw[decorate] (J.west) -- node[midway, above=7mm] {\(c_1 \colon J 
      		\mor K\)} (K.center);
      	\end{scope}	
      	\begin{scope}[decoration={brace, raise=6mm}]
      		\draw[decorate] (K.center) -- node[midway, above=7mm] {\(c_2 \colon K 
      		\mor L\)} (L.east);
      	\end{scope}	

      	\begin{scope}[->]
      		\draw (J) -- node[midway, left=2mm] {\(f_1\)}  (T);
      		\draw (J) -- node[midway, above=.1mm] {\(g_1\)} (G);
      		\draw (K) -- node[midway, above=.1mm] {\(g_1'\)} (H);
      		\draw (K) -- node[midway,above left=.5mm] {\(f_2\)} (T);
      		\draw (L) -- node[midway, above=.1mm] {\(g_2'\)} (H);
      		\draw[dashed] (G) -- node[above left=4mm] {\(\exists h_1\)} (T);
      		\draw[dashed] (H) -- node[above right=4mm] {\(\exists h_2\)} (T);
      		\draw (K) -- node[midway, above=.1mm] {\(g_2\)}(G);
      		\draw (L) -- node[midway, right=2mm] {\(f_3\)} (T);
      	\end{scope}
    \end{tikzpicture}
  \end{center}

\noindent To prove that $(i,j) \in \mathcal{C}_{T[M]}(c_1;c_2)$ is satisfied 
from the properties gained so far, we need to show that $b_3(x) = 
b_1(x) + | \{ y \in (G'\setminus g_2';j_2(L)) \mid h(y)=x 
\}|$ holds and that there exists a morphism $h \colon G' \to T$ such that the 
following diagram commutes:

  \begin{center}
    \begin{tikzpicture}[shorten >=1pt, node distance=17mm and 19mm, on grid]
      	\draw
      		node (J) at (0,0) {\(J\)}
      		node [right=of J] (G) {\(G\)}
      		node [right=of G] (K) {\(K\)}
      		node [right=of K] (H) {\(H\)}
      		node [right=of H] (L) {\(L\)}
      		node [below=of K] (G') {\(G'\)}
      		node [below=of G'] (T) {\(T\)}
      		node (PO) at (3.75,-.8) {\((PO)\)};
    
      	\begin{scope}[decoration={brace, raise=6mm}]
      		\draw[decorate] (J.west) -- node[midway, above=7mm] {\(c_1 \colon J 
      		\mor K\)} (K.center);
      	\end{scope}	
      	\begin{scope}[decoration={brace, raise=6mm}]
      		\draw[decorate] (K.center) -- node[midway, above=7mm] {\(c_2 \colon K 
      		\mor L\)} (L.east);
      	\end{scope}	

      	\begin{scope}[->]
      		\draw (J) -- node[midway, left=2mm] {\(f_1\)}  (T);
      		\draw (J) -- node[midway, above=.1mm] {\(g_1\)} (G);
      		\draw (K) -- node[midway, above=.1mm] {\(g_1'\)} (H);
      		\draw (G) -- node[midway, above=.1mm] {\(j_1\)} (G');
      		\draw (H) -- node[midway, above=.1mm] {\(j_2\)} (G');
      		\draw (L) -- node[midway, above=.1mm] {\(g_2'\)} (H);
      		\draw (G) -- node[above left=4mm] {\(h_1\)} (T);
      		\draw[dashed] (G') -- node[above left] {\(\exists h\)} (T);
      		\draw (H) -- node[above right=4mm] {\(h_2\)} (T);
      		\draw (K) -- node[midway, above=.1mm] {\(g_2\)}(G);
      		\draw (L) -- node[midway, right=2mm] {\(f_3\)} (T);
      	\end{scope}
    \end{tikzpicture}
  \end{center}
The morphism $h \colon G' \to T$ exists and is unique due to the universal 
property of pushouts. From the two equations $b_2(x) = b_1(x) + | \{ y \in (G 
\setminus g_2(K)) \mid h_1(y)=x \}|$ and $b_3(x) = b_2(x) + | \{ y \in (H 
\setminus g_2'(L)) \mid h_2(y)=x \}|$ we can derive the following equation 
which holds for all $x \in T$:
$$ b_3(x) = b_1(x) + | \{ y \in (G \setminus g_2(K)) \mid h_1(y)=x 
\}| + | \{ y \in (H \setminus g_2'(L)) \mid h_2(y)=x \}| $$
Using the results of equation $(7)$ from the previous proof direction, we 
directly can conclude that $b_3(x) = 
b_1(x) + | \{ y \in (G'\setminus g_2';j_2(L)) \mid h(y)=x 
\}|$ also holds and therefore $(i,j) \in \mathcal{C}_{T[M]}(c_1;c_2)$ holds, 
which completes this proof. \qed
\end{proof}

\begin{proposition_for}{prop:automaton functor language equivalence}{.}
  \languageEquivalenceCountingCospan
\end{proposition_for}

\begin{proof}
We will prove the following equality:
$$ G \in \mathcal{L}(T[M]) \iff 
\exists i\in I \subseteq \mathcal{C}(\emptyset),\exists j \in F
\subseteq \mathcal{C}(\emptyset) : (i,j) \in \mathcal{C}(c)$$

\noindent "$\Rightarrow$": Since
$(c \colon \emptyset \to G \leftarrow \emptyset) \in
\mathcal{L}(T[M])$ holds, there exists a legal morphism
$\varphi \colon G \to T$ and a pair of multiplicities
$(\ell,u) \in M$ such that
$\ell \leq \mathcal{B}^n_{\phi}(s_G) \leq u$ holds. Let $(i,j)$ be
$i = (f_1 \colon \emptyset \to T,0) \in I$ and
$j = (f_2 \colon \emptyset \to T, \mathcal{B}^n_{\phi}(s_G)) \in F$
which are clearly in the relation $\mathcal{C}(c)$, i.e.
$(i,j) \in \mathcal{C}(c)$ since for all $x \in T$ the equation
$\mathcal{B}^n_{\phi}(s_G)(x) = 0 + | \{ y \in (G \setminus
g_2(\emptyset)) \mid \varphi(y)=x \} | = | \{ y \in G \mid \varphi(y)
= x \} |$ holds by definition and the following diagram commutes:
  \begin{center}
    \begin{tikzpicture}[shorten >=1pt, node distance=15mm and 15mm, on grid]
      	\draw
      		node (J) at (0,0) {\(\emptyset\)}
      		node [right=of J] (G) {\(G\)}
      		node [right=of G] (K) {\(\emptyset\)}
      		node [below=of G] (T) {\(T\)};
    
      	\begin{scope}[decoration={brace, raise=6mm}]
      		\draw[decorate] (J.west) -- node[midway, above=6mm] {\(c \colon 
      		\emptyset \mor \emptyset\)} (K.east);
    
      	\end{scope}	
      	\begin{scope}[->]
      		\draw (J) -- node[midway, left=2mm] {\(f_1\)}  (T);
      		\draw (J) -- node[midway, above=.1mm] {\(g_1\)} (G);
      		\draw (G) -- node[above right=.2mm] {\(\varphi\)} (T);
      		\draw (K) -- node[midway, above=.1mm] {\(g_2\)}(G);
      		\draw (K) -- node[midway, right=2mm] {\(f_2\)} (T);
      	\end{scope}
    \end{tikzpicture}
  \end{center}

\noindent "$\Leftarrow$": There exists $i \in I 
\subseteq \mathcal{C}(\emptyset)$ and $j \in F \subseteq 
\mathcal{C}(\emptyset)$ with $i = (f_1 \colon \emptyset \to T,0)$ and $j = (f_2 
\colon \emptyset \to T, b)$ such that $(i,j) \in \mathcal{C}(c)$ holds. 
Therefore, there exists a pair of multiplicities $(\ell,u) \in M$ with $\ell 
\leq b \leq u$ and we get that there exists a morphism $\varphi \colon G \to 
T$ such that the following diagram commutes:
  \begin{center}
    \begin{tikzpicture}[shorten >=1pt, node distance=15mm and 15mm, on grid]
      	\draw
      		node (J) at (0,0) {\(\emptyset\)}
      		node [right=of J] (G) {\(G\)}
      		node [right=of G] (K) {\(\emptyset\)}
      		node [below=of G] (T) {\(T\)};
    
      	\begin{scope}[decoration={brace, raise=6mm}]
      		\draw[decorate] (J.west) -- node[midway, above=6mm] {\(c \colon 
      		\emptyset \mor \emptyset\)} (K.east);
    
      	\end{scope}	
      	\begin{scope}[->]
      		\draw (J) -- node[midway, left=2mm] {\(f_1\)}  (T);
      		\draw (J) -- node[midway, above=.1mm] {\(g_1\)} (G);
      		\draw[dashed] (G) -- node[above right=.2mm] {\(\exists \varphi\)} (T);
      		\draw (K) -- node[midway, above=.1mm] {\(g_2\)}(G);
      		\draw (K) -- node[midway, right=2mm] {\(f_2\)} (T);
      	\end{scope}
    \end{tikzpicture}
  \end{center}
For all $x \in T$ the following equation holds:
\begin{align*}
  b(x) &= 0 + | \{ y \in (G \setminus g_2(\emptyset)) \mid \varphi(y)=x \} | \\
       &= | \{ y \in G \mid \varphi(y)=x \} | \\
       &= \mathcal{B}^n_{\phi}(s_G)(x)
\end{align*}
From $\ell \leq b \leq u$ we can infer that $\varphi \colon G \to T$
is a legal morphism due to the fact that
$\ell \leq \mathcal{B}^n_{\phi}(s_G) \leq u$ holds as well, and
therefore $G \in \mathcal{L}(T[M])$. \qed
\end{proof}

\begin{proposition_for}{prop:multi annot graphs are closed under intersect}{.}
  \closurePropsAnnotIntersect
\end{proposition_for}

\begin{proof}
  Let two multiply annotated type graphs $T_1[M_1]$ and $T_2[M_2]$ be
  given. Let $T_1\times T_2$ be the usual product graph in the
  underlying category $\GR$.

  We now consider the multiply annotated type graph
  $(T_1 \times T_2)[N]$ where the set of annotations $N$ is defined as
  follows:
  \begin{align*}
    N = \{(\ell,u) \mid &\ \ell,u \in \mathcal{A}(T_1\times T_1)\text{
      such that} \\ &\ \pi_1
    \colon (T_1\times T_2)[\ell,u] \to T_1[M_1] \text{ is legal and } \\
    &\ \pi_2 \colon (T_1 \times T_2)[\ell,u] \to T_2[M_2] \text{ is legal}
    \}
  \end{align*}  

  Therefore for each $(\ell,u) \in N$ there exist
  $(\ell_1,u_1) \in M_1$ and $(\ell_2,u_2) \in M_2$ such that the
  following four properties hold:

  \noindent \begin{minipage}{0.45\textwidth}
    \begin{align*}
      \mathcal{A}_{\pi_1}(\ell) &\geq \ell_1 \qquad\
      \mathcal{A}_{\pi_1}(u) \leq u_1
      \\
      \mathcal{A}_{\pi_2}(\ell) &\geq \ell_2 \qquad
      \mathcal{A}_{\pi_2}(u) \leq u_2
      \\
    \end{align*} \\

  \end{minipage} 
  \begin{minipage}{0.58\textwidth}
    \vspace*{-.9cm}
    \scalebox{.9}{
      \xymatrix{
        & (T_1 \times T_2)[N] \ar@{->}[dl]^{\pi_1} 
        \ar@{->}[dr]_{\pi_2} & \\
        T_1[M_1] & & T_2[M_2]
      }
    }
  \end{minipage}
  \vspace*{-.9cm}
  
  \noindent We will now prove the following equality:
  \[
    \mathcal{L}(T_1[M_1]) \cap \mathcal{L}(T_2[M_2]) = 
    \mathcal{L}((T_1 \times T_2)[N])
  \] 
  
  \noindent $\subseteq$: Let
  $G \in \mathcal{L}(T_1[M_1]) \cap \mathcal{L}(T_2[M_2])$. Then there
  exist two legal morphisms $\varphi_1 \colon G[s_G,s_G] \to T_1[M_1]$
  and $\varphi_2 \colon G[s_G,s_G] \to T_2[M_2]$. Due to the universal
  property of pullbacks in the underlying category $\GR$, there exists
  a unique graph morphism $\eta \colon G \to T_1 \times T_2$ such that
  the following diagram commutes:

  \scalebox{.9}{ \hspace*{3cm} \xymatrix{ & G[s_G,s_G]
      \ar@{->}@/_1pc/[ddl]_{\varphi_1}
      \ar@{->}@/^1pc/[ddr]^{\varphi_2}
      \ar@{.>}[d]^{\eta} & \\
      & (T_1 \times T_2)[N] \ar@{->}[dl]^{\pi_1}
      \ar@{->}[dr]_{\pi_2} & \\
      T_1[M_1] & & T_2[M_2] } } \vspace*{.5cm}

  Since $\varphi_i = \pi_i \circ \eta$ with $i \in \{1,2\}$ is a legal
  morphism, there exist annotations $(\ell_1,u_1) \in M_1$ and
  $(\ell_2,u_2) \in M_2$ such that the following inequalities hold:
  \begin{align*}
    \ell_1 \leq \mathcal{A}_{\phi_1}(s_G) = \mathcal{A}_{\pi_1 \circ
      \eta}&(s_G)
    = \mathcal{A}_{\pi_1}(\mathcal{A}_{\eta}(s_G)) \leq u_1 \\
    \ell_2 \leq \mathcal{A}_{\phi_2}(s_G) = \mathcal{A}_{\pi_2 \circ
      \eta}&(s_G) = \mathcal{A}_{\pi_2}(\mathcal{A}_{\eta}(s_G)) \leq
    u_2
  \end{align*}
  Therefore the pair
  $(\mathcal{A}_{\eta}(s_G),\mathcal{A}_{\eta}(s_G))$ is one of the
  annotations in $N$ and we can conclude that 
  $G \in \mathcal{L}((T_1 \times T_2)[N])$ holds.\\

  \noindent $\supseteq$: We now assume
  $G \in \mathcal{L}((T_1 \times T_2)[N])$. Then there exists a legal
  morphism $\eta \colon G[s_G,s_G] \to (T_1 \times T_2)[N]$ with an
  annotation pair $(\ell,u) \in N$ such that
  $\ell \leq \mathcal{A}_{\eta}(s_G) \leq u$. For each such pair
  $(\ell,u) \in N$ we have two legal morphisms
  $\pi_1 \colon (T_1 \times T_2)[\ell,u] \to T_1[M_1]$ and
  $\pi_2 \colon (T_1 \times T_2)[\ell,u] \to T_2[M_2]$, by
  construction. We obtain two morphisms
  $\varphi_1 \colon G[s_G,s_G] \to T_1[M_1]$ with
  $\varphi_1 = \pi_1 \circ \eta$ and
  $\varphi_2 \colon G[s_G,s_G] \to T_2[M_2]$ with
  $\varphi_2 = \pi_2 \circ \eta$, which are legal due to Lemma
  \ref{le:composition of legal morphism}. Therefore we can conclude
  that $G \in (\mathcal{L}(T_1[M_1]) \cap
  \mathcal{L}(T_2[M_2]))$. \qed
\end{proof}

In order to show closure under union for  annotated type
graphs  over  $\mathcal{B}^n$, we first have a look at the following lemma.

\begin{lemma}
  \label{lem:property-union}
  Assume that we are working with annotations  over  $\mathcal{B}^n$.\\[0.1cm]
      \begin{minipage}[c]{.67\linewidth}
        Let $i \colon A[M] \to T[N]$ and
        $\varphi \colon G[s_G,s_G] \to T[N]$ be two legal graph
        morphisms where $i$ is injective. Let $(\ell,u) \in M$ be one
        of the double multiplicities of the graph $A$. Whenever
        $\mathcal{B}^n_{\varphi}(s_G) \leq \mathcal{B}^n_{i}(u)$, we can
        deduce that there exists a graph morphism
        $\zeta \colon G \to A$ with $i \circ \zeta = \varphi$,
        i.e. the diagram commutes.
      \end{minipage}
      \quad
      \begin{minipage}[c]{.28\linewidth}
        \begin{center}
          \scalebox{.9}{ 
            \xymatrix{
                 & G[s_G,s_G] \ar@{->}[d]^{\varphi} \ar@{.>}[dl]_{\zeta} \\
                 A[M] \ar@{>->}[r]_{i} & T[N] \\
                     }
          } 
        \end{center}
      \end{minipage} 
\end{lemma}

\begin{proof}
  The morphisms $\zeta$ exists if all elements of the form $\phi(x)$
  with $x\in G$ are in the range of $i$. For such an $x$ we have
  $1 = s_G(x) \le \mathcal{B}^n_\phi(s_G)(\phi(x))$, since
  $\mathcal{B}^n_\phi(s_G)(\phi(x))$ is the sum of the $s_G$-annotations
  of all preimages of $x$. Furthermore
  $\mathcal{B}^n_\phi(s_G)(\phi(x)) \le
  \mathcal{B}^n_i(u)(\phi(x))$. But $\mathcal{B}^n_i(u)(y) = 0$ for
  all $y\in T$ that are not in the range of $i$, since the empty sum
  evaluates to $0$. But since $\mathcal{B}^n_i(u)(\phi(x)) \ge 1$, we
  can conclude that $\phi(x)$ has a preimage under $i$.  \qed
\end{proof}

In addition, we need the concept of reduction: the reduction operation
shifts annotations over morphisms in the reverse direction.

\begin{definition}[Reduction]
  \label{def:reduction}
  Let $\mathcal{A}$ be an (annotation) functor.  For a morphism
  $\phi\colon G\to G'$ and a monoid element $a'\in \mathcal{A}(G')$ we
  define the reduction of $a'$ to $G$ as follows:
  \[ \mathit{red}_\phi(a') = \bigvee \{a\mid \mathcal{A}_\phi(a) \le
    a'\}. \] 
\end{definition}

In the case of concrete annotations,  the reduction operator satisfies
the following properties:

\begin{lemma}
  \label{lem:properties-red}
  Assume that we are working with annotations over
  $\mathcal{B}^n$. 
If $\varphi \colon G \to H$ is injective, we
  obtain the following equality for all $x \in G$:
  \[ \mathit{red}_\phi(a')(x) = a'(\varphi(x)) \] Furthermore,
  if $\phi\colon G\to G'$ is injective, it holds that
  $\mathit{red}_\phi(\mathcal{B}^n_\phi(a)) = a$ for every
  $a\in\mathcal{B}^n(G)$.
\end{lemma}

\begin{proof}
  Straightforward from the definition of concrete annotations. \qed
\end{proof}



We are now ready to prove closure under union for the concrete
case. Since we do work with abstract annotations in the proof, but
need the results of the lemmas, one could generalize this result to a
setting where the properties stated in Lemma~\ref{lem:property-union}
and Lemma~\ref{lem:properties-red} hold.

\smallskip

\begin{proposition_for}{prop:multi annot graphs are closed under union}{.}
  \closurePropsAnnotUnion
\end{proposition_for}

\begin{proof}
  Let two multiply annotated type graphs $T_1[M_1]$ and $T_2[M_2]$ be
  given. Let $T_1 \oplus T_2$ be the usual coproduct graph in the
  underlying category $\GR$, together with the embedding morphisms
  $i_1 \colon T_1 \to T_1 \oplus T_2$ and $i_2 \colon T_2 \to
  T_1 \oplus T_2$: \\

  \scalebox{.9}{ \hspace*{3cm}
    \xymatrix{
      T_1[M_1] \ar@{>->}[dr]_{i_1} & & T_2[M_2] \ar@{>->}[dl]^{i_2} \\
      & T_1 \oplus T_2[N] & \\
    }
  } \vspace*{.5cm}
  
  \noindent We define the set of annotations $N$ for the multiply
  annotated type graph $(T_1 \oplus T_2)[N]$ using the following two
  sets:
  \begin{align*}
    N_1 &= \{\ (\mathcal{B}^n_{i_1}(\ell_1),\mathcal{B}^n_{i_1}(u_1))\ \mid \ 
    (\ell_1,u_2) \in M_1 \} \\
    N_2 &= \{(\mathcal{B}^n_{i_2}(\ell_2),\mathcal{B}^n_{i_2}(u_2)) \mid \  
    (\ell_2,u_2) \in M_2 \}
  \end{align*}  

  Finally we define $N = N_1 \cup N_2$. \\

  By this definition, we get that for all elements $x \in T_1$ and for
  all $(\ell,u) \in M$ there exists $(\ell_1,u_1) \in N_{T_1(M)}$ such
  that $\mathcal{B}^n_{i_1}(\ell)(i_1(x)) = \ell_1(i_1(x))$ and
  $\mathcal{B}^n_{i_1}(u)(i_1(x)) = u_1(i_1(x))$. This makes $i_1$ a
  legal morphism since $N_1 \subseteq N$. The same holds for $i_2$
  analogously.  We will now prove the following equality:
  \[
    \mathcal{L}(T_1[M_1]) \cup \mathcal{L}(T_2[M_2]) =
    \mathcal{L}((T_1 \oplus T_2)[N])
  \] 
  
  $\subseteq$: Let
  $G \in (\mathcal{L}(T_1[M_1]) \cup \mathcal{L}(T_2[M_2]))$. Then there
  exists at least one legal morphism
  $\varphi_1 \colon G[s_G,s_G] \to T_1[M_1]$ or
  $\varphi_2 \colon G[s_G,s_G] \to T_2[M_2]$. We assume that
  $G \in \mathcal{L}(T_1[M_1])$. Let
  $\eta \colon G[s_G,s_G] \to (T_1 \oplus T_2)[N]$ be the composed
  morphism of $i_1$ and $\varphi_1$ with $\eta = i_1 \circ
  \varphi_1$. Then $\eta$ is legal due to Lemma \ref{le:composition of
    legal morphism} and therefore $G \in \mathcal{L}((T_1 \oplus T_2)[N])$
  holds.  The proof for the case where $G \in \mathcal{L}(T_2[M_2])$
  works in the same way.

  $\supseteq$: We now assume $G \in \mathcal{L}((T_1 \oplus
  T_2)[N])$. Then, there exists a legal morphism
  $\eta \colon G[s_G,s_G] \to (T_1 \oplus T_2)[N]$ with an annotation
  $(\ell,u) \in N$ such that
  $\ell \leq \mathcal{B}^n_{\eta}(s_G) \leq u$. For each
  $(\ell,u) \in N$, we know that the pair belongs to $N_1$ or
  $N_2$. Assume that $(\ell,u) \in N_1$. Then we know that there
  exists $(\ell_1,u_1)\in M_1$ such that
  $\ell=\mathcal{B}^n_{i_1}(\ell_1)$, $u=\mathcal{B}^n_{i_1}(u_1)$.
  Hence $\mathcal{B}^n_\eta(S_G) \le u = \mathcal{B}^n_{i_1}(u_1)$.
  From Lemma~\ref{lem:property-union} it follows that there exists a
  graph morphism $\zeta_1 \colon G \to T_1$ with
  $\eta = i_1 \circ \zeta_1$ such that the following diagram commutes
  in the underlying category $\GR$:

  \scalebox{.9}{ \hspace*{3cm}
    \xymatrix{
         & G[s_G,s_G] \ar@{.>}[dl]^{\zeta_1} 
         \ar@{->}[dd]^{\eta} & \\
         T_1[M_1] \ar@{>->}[dr]_{i_1} & & T_2[M_2] \ar@{>->}[dl]^{i_2} \\
         & (T_1 \oplus T_2)[N] & \\
             }
  } \vspace*{.5cm}
 
  We need to prove that $\zeta_1$ is a legal graph morphism in the
  category of multiply annotated graphs. We get that
  $\mathcal{B}^n_{\eta}(s_G) = \mathcal{B}^n_{i_1 \circ \zeta_1}(s_G) =
  \mathcal{B}^n_{i_1}(\mathcal{B}^n_{\zeta_1}(s_G))$ and since $i_1$ is
  injective, the following inequality holds due to the fact that
  $\mathit{red}_\phi$ is monotone and
  $\mathit{red}_\phi(\mathcal{B}^n_\phi(a)) = a$ holds for every
  $a\in\mathcal{B}^n(G)$, whenever $\phi$ is injective
  (cf. Lemma~\ref{lem:properties-red}):

  \begin{align*}
    && \mathcal{B}^n_{i_1}(\ell_1) &\leq \quad \qquad
    \mathcal{B}^n_{\eta}(s_G) \quad \,\,\, \quad \leq \mathcal{B}^n_{i_1}(u_1) 
    \\
    \Rightarrow && \mathcal{B}^n_{i_1}(\ell_1) &\leq \qquad
    \mathcal{B}^n_{i_1}(\mathcal{B}^n_{\zeta_1}(s_G)) \quad \leq
    \mathcal{B}^n_{i_1}(u_1)
    \\
    \Rightarrow && \mathit{red}_{i_1}(\mathcal{B}^n_{i_1}(\ell_1)) &\leq
    \mathit{red}_{i_1}(\mathcal{B}^n_{i_1}(\mathcal{B}^n_{\zeta_1}(s_G)))
    \leq \mathit{red}_{i_1}(\mathcal{B}^n_{i_1}(u_1))
    \\
    \Rightarrow && \ell_1 &\leq \quad \qquad \mathcal{B}^n_{\zeta_1}(s_G)
    \quad \,\,\, \quad \leq u_1
    \\
  \end{align*}

  Therefore $\zeta_1 \colon G[s_G,s_G] \to T_1[M_1]$ is a legal
  morphism and we can conclude that $G \in \mathcal{L}(T_1[M_1])$. For
  a legal morphism $\eta \colon G[s_G,s_G] \to (T_1 \oplus T_2)[N]$
  with a pair $(\ell,u) \in N_2$ we get a similar proof which shows
  that $G \in \mathcal{L}(T_2[M_2])$. Summarizing, in all cases
  $G \in (\mathcal{L}(T_1[M_1]) \cup \mathcal{L}(T_2[M_2]))$ holds.
\end{proof}

\newpage
\section{Extended Example: Annotated Type Graphs}
\label{sec:exampleAnnotGraphs}


In order to illustrate the use of annotated type graphs in
applications, we model a client-server scenario with the following
specification:

\begin{itemize}
  \item There exists exactly one server.
  \item An arbitrary number of users can connect to the server, even using 
    multiple connection sessions at the same time.
  \item There exists one user with special administrative rights.
  \item At least one user is always connected to the server.
  \item The server can host an arbitrary number of files from which at
    most one can be edited at the same time.
\end{itemize}

\noindent The above scenario can be modelled using an annotated type
graph $T_1[\ell,u]$ (see below). We will use the following edge
labels: $A$-labeled loops for administrative rights, $C$-labeled
edges for connections between users and the server and $E$-labeled
edges which are pointing to the file that is currently edited. We now
extend the requirements of our specification:

\begin{itemize}
\item The user with the administrative rights is always connected to
  the server.
  \item There has to be at least one file on the server.
\end{itemize}

\noindent We use the annotated type graph $T_2[\ell',u']$, depicted
below to model the extended scenario.  


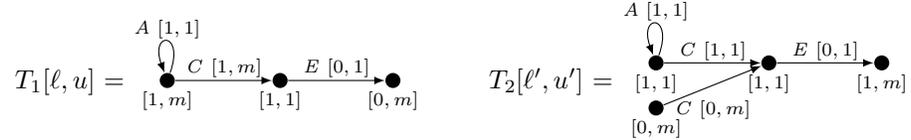
\begin{center}
  \begin{tabular}{ccccc}
    $T_1[\ell,u]$ = &
    \begin{tikzpicture}[x=1.5cm,y=-1.2cm,baseline=(1.south)]
      \node[glab] (top) at (0,-.1) {} ;
      \node[gnode] (1) at (0,0) {} ; 
      \node[glab,below] (lab1) at (1.south) {$[1,m]$} ;
      \node[gnode] (2) at (1,0) {} ; 
      \node[glab,below] (lab2) at (2.south) {$[1,1]$} ;
      \node[gnode] (3) at (2,0) {} ; 
      \node[glab,below] (lab3) at (3.south) {$[0,m]$} ;
         	    \draw[gedge] (1) .. controls +(70:.7cm) and +(110:.7cm) .. 
         	            node[arlab,above] (labm3) {$\mathit{A}\ [1,1]$} (1) 
         	            ;
      \draw[gedge] (1) to node[arlab,above] {$\mathit{C}\ [1,m]$} (2) ;
      \draw[gedge] (2) to node[arlab,above] {$\mathit{E}\ [0,1]$} (3) ;
    \end{tikzpicture}
  & \hspace*{.7cm}
   &
    $T_2[\ell',u']$ = &
    \begin{tikzpicture}[x=1.5cm,y=-1.2cm,baseline=(base.south)]
      \node[glab] (base) at (0,.75) {} ;
      \node[glab] (top) at (0,-.1) {} ;
      \node[gnode] (1) at (0,0.5) {} ; 
      \node[glab,below] (lab1) at (1.south) {$[1,1]$} ;
      \node[gnode] (2) at (1,0.5) {} ; 
      \node[glab,below] (lab2) at (2.south) {$[1,1]$} ;
      \node[gnode] (3) at (2,0.5) {} ; 
      \node[glab,below] (lab3) at (3.south) {$[1,m]$} ;
      \node[gnode] (5) at (0,1) {} ; 
      \node[glab,below] (lab5) at (5.south) {$[0,m]$} ;
         	    \draw[gedge] (1) .. controls +(70:.7cm) and +(110:.7cm) .. 
         	            node[arlab,midway,above] (labm3) {$\mathit{A}\ 
         	            [1,1]$} (1) 
         	            ;
      \draw[gedge] (1) to node[arlab,above] {$\mathit{C}\ [1,1]$} (2) ;
      \draw[gedge] (5) to node[arlab,below=.15cm] {$\mathit{C}\ [0,m]$} 
      (2) ;
      \draw[gedge] (2) to node[arlab,above] {$\mathit{E}\ [0,1]$} (3) ;
    \end{tikzpicture}
  \end{tabular}
\end{center}

\noindent Since the second scenario is more restrictive than the
first, there exist graphs in $\mathcal{L}(T_1[\ell,u])$, which do not
fulfil the additional requirements of the extended specification.

  \begin{wrapfigure}{r}{2.2cm}
  \vspace{-1.2cm}
    \begin{tikzpicture}[x=1.5cm,y=-1.2cm,baseline=(base.south)]
      \node[glab] (base) at (0,.75) {} ;
      \node[glab] (top) at (0,-.1) {} ;
      \node[gnode] (1) at (0,0.5) {} ; 
      \node[gnode] (2) at (1,0.5) {} ; 
      \node[gnode] (5) at (0,1) {} ; 
         	    \draw[gedge] (1) .. controls +(70:.7cm) and +(110:.7cm) .. 
         	         node[arlab,midway,above] (labm3) {$\mathit{A}$} (1) ;
      \draw[gedge] (5) to node[arlab,below=.15cm] {$\mathit{C}$} 
      (2) ;
    \end{tikzpicture}
    \vspace{-.9cm}
  \end{wrapfigure}

  \noindent For instance the graph $G$ shown to the right is such a
  model, which describes that there exists a user with administrative
  rights but he is not connected to the server. Instead there is
  another user which is currently connected. However, it holds that
  that
  $\mathcal{L}(T_2[\ell',u']) \subseteq \mathcal{L}(T_1[\ell,u])$,
  since we can easily find a legal graph morphism
  $\varphi \colon T_2[\ell',u'] \to T_1[\ell,u]$.
  
}
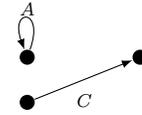
\end{document}